\documentclass[12pt,fleqn]{article}

\usepackage{amsfonts,url}
\usepackage[normalem]{ulem}
\usepackage{amssymb}
\usepackage{amsmath}
\usepackage{amstext,verbatim,graphicx,ifthen,multirow,amsthm}
\usepackage{bm}
\usepackage{natbib}
\usepackage[bf]{caption}
\usepackage[nodisplayskipstretch]{setspace}
\usepackage{mathtools}
\usepackage{booktabs}
\usepackage{threeparttable}
\usepackage{tikz}
\usetikzlibrary{arrows.meta}
\usepackage{xfrac}
\usepackage{subcaption}
\usepackage{appendix}
\usepackage{adjustbox}

\usepackage[pdftex,unicode,bookmarksnumbered=false, hyperfootnotes=true]{hyperref}

\usepackage{color-edits}

\addauthor{da}{red}
\usepackage{xcolor}
\addauthor{gwi}{orange}

\renewcommand{\mathbf}{\boldsymbol}
\renewcommand{\thepage}{}
\newcommand{\mmv}{\mathbb{V}}

\newcommand{\numcl}{M}
\newcommand{\htauf}{\hat{\tau}^{\rm fe}}

\newcommand{\indep}{\perp\!\!\!\perp}
\newcommand{\obi}{X_i,\os_g}

\newcommand{\ova}{\overline{A}}
\newcommand{\uw}{\underline{W}}
\newcommand{\ux}{\underline{X}}

\newcommand{\mme}{\mathbb{E}}
\newcommand{\mma}{\mathbb{A}}

\newcommand{\os}{\overline{S}}

\newcommand{\ow}{\overline{W}}

\newcommand{\owx}{\overline{WX}}

\newcommand{\ox}{\overline{X}}

\newcommand{\pr}{{\rm pr}}

\newcommand{\ou}{\overline{U}}
\newcommand{\oui}{\overline{U}_{g(i)}}
\newcommand{\osi}{\overline{S}_{g(i)}}

\newtheorem{corollary}{Corollary}
\newtheorem{assumption}{Assumption}[section]
\newtheorem{theorem}{Theorem}
\newtheorem{lemma}{Lemma}
\newtheorem{prop}{Proposition}

\def\monthname{\ifcase\month\or
  January\or February\or March\or April\or May\or June\or July\or
  August\or September\or October\or November\or December\fi}
\numberwithin{equation}{section}

\def\monthname{\ifcase\month\or
January\or February\or March\or April\or May\or June\or
July\or August\or September\or October\or November\or December\fi}
\renewcommand{\appendix}{\small\parindent 0cm\setcounter{equation}{0}
\renewcommand{\theequation}{A.\arabic{equation}}
\setcounter{lemma}{0}\renewcommand{\thelemma}{A.\arabic{lemma}}
\setcounter{theorem}{0}\renewcommand{\thetheorem}{A.\arabic{theorem}}}
\begin{document}

\title{\textbf{
Fixed Effects and the Generalized Mundlak Estimator}\thanks{{\small We are grateful for comments by seminar participants at Harvard, MIT, Princeton, Brown, NYU, the SIEPR lunch at Stanford,  the International Association of Applied Econometrics meeting in Montreal, Manuel Arellano, Matias Cattaneo, Pat Kline,  and Michael Kolesar. We are grateful to 
Nicola Lacetera, Mario Macis, and Robert Slonim for sharing their data and making it publicly available.  We are also grateful to Greg Duncan for raising questions this paper tries to answer. The data and code underlying this research is available on Zenodo at https://doi.org/10.5281/zenodo.8226636. This research was generously supported by ONR grants N00014-17-1-2131 and N00014-19-1-2468 and a gift from Amazon. A previous version of this manuscript  circulated under the title ``The Role of the Propensity Score in Fixed Effect Models'' \citep{arkhangelsky2018role}.}} }
\author{Dmitry  Arkhangelsky \thanks{{\small  Associate Professor, CEMFI, CEPR darkhangel@cemfi.es. }} \and Guido W. Imbens\thanks{{\small Professor of
Economics,
Graduate School of Business and Department of Economics, Stanford University, SIEPR, and NBER,
imbens@stanford.edu.}} }
\date{\ifcase\month\or
January\or February\or March\or April\or May\or June\or
July\or August\or September\or October\or November\or December\fi \ \number
\year\ \ }
\maketitle\thispagestyle{empty}

\begin{abstract}
\singlespacing
\noindent We develop a new approach for estimating average treatment effects in observational studies with unobserved group-level heterogeneity. We consider a general model with group-level unconfoundedness and provide conditions under which aggregate balancing statistics -- group-level averages of functions of treatments and covariates -- are sufficient to eliminate differences between groups. Building on these results, we reinterpret commonly used linear fixed-effect regression estimators by writing them in the Mundlak form as linear regression estimators without fixed effects but including group averages. We use this representation to develop Generalized Mundlak Estimators (GMEs) that capture group differences through group averages of (functions of) the unit-level variables and adjust for these group differences in flexible and robust ways in the spirit of the modern causal literature.
\end{abstract}

\noindent \textbf{Keywords}: fixed effects, cross-section data, groups, causal effects, treatment effects, unconfoundedness.

\begin{center}
\end{center}



\baselineskip=20pt
\setcounter{page}{1}
\renewcommand{\thepage}{\arabic{page}}
\renewcommand{\theequation}{\arabic{section}.\arabic{equation}}

\section{Introduction}

A common specification for regression functions when estimating causal effects with grouped  data is a fixed effect model:
\begin{equation}
\label{een}
Y_i=\alpha_{g(i)}+W_i\tau+X_i^\top\beta+\varepsilon_i,
\end{equation}
where the index $g(i)$ indicates the group  a unit $i$ belongs to, 
the $\alpha_g$  are the group  fixed effects, and $\varepsilon_i$ is an error term, independent of the regressors $W_i$ and $X_i$  (see  \citet*{
arellano2003panel, angrist2008mostly,wooldridge2010econometric} for  textbook discussions).
 In empirical work the groups (also referred to as strata, subpopulations, or clusters) may correspond to states, cities, MSAs, classrooms, birth cohorts, families, siblings, or other geographic or demographic groups.  
The model is typically  estimated by least squares. We refer to the corresponding estimator as the fixed effect estimator. The parameter $\tau$, which we would like to interpret as an average causal effect of the treatment  $W_i$ (binary through most of this paper), is the object of interest. The fixed effects $\alpha_g$  are intended to capture unobserved differences between the groups. The motivation for including these fixed effects in the specification is  that their presence improves the credibility of a causal interpretation of the fixed effect estimator for $\tau$ ({\it e.g.,} \citet*{arellano2003panel, angrist2008mostly}). Inference for $\tau$ is typically based on asymptotic approximations with  a growing number of groups and a fixed number of units per group. In that setting one cannot rely on consistent  estimation of the fixed effects because of the incidental parameter problem ({\it e.g.,} \citet*{neyman1948consistent, lancaster2000incidental}). 

In this paper we make two sets of contributions. First, we  unpack the assumptions underlying the fixed effects estimator based on specification (\ref{een}), and second, we re-interpret the fixed effect estimator and use that re-interpretation to propose a new class of estimators.

The popularity of the  fixed effect estimator may be due partly because it is viewed as accounting for differences between groups in a flexible way. However, the specification in (\ref{een}) embodies multiple strong assumptions. The current paper clarifies these assumptions by separating them  into three parts. First, the specification in  (\ref{een}) assumes a constant treatment effect. In practice, it is likely the effects of the treatment are heterogeneous.  We show that in general the average effect of the treatment is {\em not} point-identified. We then describe the estimand corresponding to the fixed effect estimator under general treatment-effect heterogeneity and demonstrate that it estimates, under some conditions, a weighted average treatment effect, with the weights depending in an unusual way on the sampling scheme. For example, we show that in the presence of treatment effect heterogeneity the interpretation of the fixed effect estimator changes if we double the number of units we sample per cluster.  Second,  a key assumption motivating (\ref{een}),
what we call group unconfoundedness (Assumption \ref{as:unconf_cl_un} below), validates all within-group comparisons of treated and control units with the same covariates. However, the additional assumptions
underlying (\ref{een})  also implicitly validate some, but not all, between-group comparisons of treated and control units with the same covariates. We describe which between-group comparisons are validated by the fixed effect specification. Third, the fixed effect specification   assumes linearity and additivity in the fixed effects, covariates, and treatments.

In the second set of contributions, we propose a new class of estimators.
To motivate these new estimators we  start with the re-interpretation of the fixed effect estimator due to \cite{mundlak1978pooling}.
Instead of thinking of the fixed effect estimator for $\tau$ as corresponding to the specification in (\ref{een}), we can also think of it as the least squares estimator for $\tau$
corresponding to the specification
\begin{equation}
\label{twee0}
Y_i=W_i\tau+ \overline{W}_{g(i)}\gamma+X_i^\top\beta+
\overline{X}_{g(i)}^\top\delta+\eta_i,
\end{equation}
where, throughout the paper, we use $\overline{W}_g$ and $\overline{X}_g$ to denote the average value of $W_i$ and $X_i$ for units in group $g$, that is, for units with $g(i)=g$.
Compared to (\ref{een}), this specification does not have the fixed effects $\alpha_g$. Instead, it has two additional regressors, the group averages of $W_i$ and $X_i$. The least squares estimators for $\tau$ based on (\ref{een}) and (\ref{twee0}) are identical, as first noted by \cite{mundlak1978pooling} (see also \cite{wooldridge2021two}).
We can therefore interpret the fixed effect estimator as controlling for group differences by simply including group averages of the covariates as additional regressors in a linear regression. This interpretation motivates our exploration of new estimators, which we refer to as Generalized Mundlak Estimators (GMEs). These estimators, for suitable defined average treatment effects  as made precise in Section \ref{section:identification}, maintain group unconfoundedness (Assumption \ref{as:unconf_cl_un}) which justifies within-group comparisons of treated and control units. We then formulate additional assumptions that allow for some, but not all, cross-group comparisons of treated and control units. The estimators can be thought of as starting with a parsimonious specification as in  (\ref{twee0}), and, like (\ref{twee0}), are based on adjusting for group characteristics to justify cross-group comparisons. They differ from (\ref{twee0}) in that they free up the specification or change the estimation procedure in three distinct ways, inspired by the modern literature on estimating average treatment effects under unconfoundedness ({\it e.g.,} \citet{imbens2015causal, abadie2018econometric}). 
First, we can adjust for differences in the group averages in a more flexible, nonlinear way, for example, by including nonlinear functions of the averages $\overline{W}_{g}$ and $\overline{X}_{g}$ in the regression function in ({\ref{twee0})}. Second, we can make the estimator more robust by  estimating the propensity score and using inverse propensity score weighting in combination with regression to create doubly robust estimators (\citet*{robins1995semiparametric,  zubizarreta2015stable, chernozhukov2017double, athey2018approximate}). Third, we can account for more general differences between groups than accounted for by the fixed effect estimator by adjusting for differences in averages of other functions of the covariates $X_i$  and the treatment indicators $W_i$, for example, the average of the product, $\overline{WX}_{g}$.
We show how the choice of group characteristics to adjust for can be motivated by structural choice models in Section \ref{sec:discussion}.

We can capture the three modifications in the proposed Generalized Mundlak Estimator relative to the fixed effect estimators by an unconfoundedness assumption that strengthens group unconfoundedness. Instead of conditioning on group indicators, it assumes it is sufficient to condition on a set of group characteristics:
\[W_i\ \indep\ \Bigl(Y_i(0),Y_i(1)\Bigr)\ \Bigl|\ X_i,
\overline{H}_{g(i)},\]
where $Y_i(0)$ and $Y_i(1)$ are potential outcomes for unit $i$ (\citet*{neyman1923,rubin1974estimating, imbens2015causal}), and $\overline{H}_{g(i)}$ is the
group average  of some pre-specified function $H(W_i,X_i)$  of the treatment indicators and the covariates. We refer to $\overline{H}_{g(i)}$ as a balancing score, following the terminology in \citet{rosenbaum1983central}. By the unconfoundedness condition, treated and control units with the same values for covariates and balancing statistics can be compared even if they belong to different groups.
If $H(W_i,X_i)=(W_i,X_i)$ and the adjustment is only linear, this gets us back to the fixed effect estimator, but in general the proposed Generalized Mundlak Estimator  differs in two ways. First, it adjusts more flexibly for the components of $\overline{H}_{g(i)}$, and second,  $H(\cdot)$ can contain more components than just $(W_i,X_i)$.

Our formal asymptotic results focus on the case with a fixed number of units per group. In particular, we discuss settings where the group size is not  large enough to carry out a two-stage procedure where we first estimate the effects entirely within groups by flexibly adjusting for covariates, followed by averaging over the groups. In other words, because of the incidental parameter problem (\citet*{neyman1948consistent, lancaster2000incidental}), we need to rely on comparisons of treated and control units in different groups. In addition, there is  concern that  accounting for the group differences solely  through additive fixed effects as in specification (\ref{een}) is not sufficient to adjust for all relevant differences ({\it e.g.,} \citet*{altonji2005cross,imai2019should}). Finally, in this setting we cannot estimate the propensity score, that is, the conditional probability of treatment assignment, as a function of group membership.

 \section{Examples}\label{section:simple}
 In this section, we discuss the main insights of the paper in the context of two examples. We start with a case with no covariates and present the interpretation of (\ref{een}) under general heterogeneity in treatment effects. We then introduce binary covariates and connect (\ref{een}) to a particular unconfoundedness condition. We finish the section with a set of practical recommendations.
 
\subsection{No covariates}
 
 A key assumption underlying most group analyses is that assignment is unconfounded within the groups. Let each unit $i$ in a large population be characterized by a pair of potential outcomes  $(Y_i(0),Y_{i}(1))$ and group label $L_{g(i)}$.  These labels can correspond to names, \textit{e.g.}, geographic locations, legal entities, and markets. We use the labels at this stage to make clear that there is not necessarily an ordering to the groups, nor a distance measure that allows us to say that some groups are closer to each other. Suppose each unit is assigned to a binary treatment $W_i$.
  The sampling process has two stages. First, we randomly sample $\numcl$ groups from a large population of groups. Second, we sample $N_g$ units from group $g$, for each of the sampled groups. In the absence of individual-level covariates, we can express group unconfoundedness in the following way (\citet*{rosenbaum1983central}):
 \begin{equation}\label{eq:sim_exp_unc_clusters_old}
W_i\ \indep\ \Bigl(Y_i(0),Y_i(1)\Bigr)\ \Bigl|\ L_{g(i)}.
\end{equation}
This assumption validates comparisons between treated and control units within groups.

As we show formally in the next section, this condition implies a different independence restriction:
 \begin{equation}\label{eq:sim_exp_unc_clusters_w}
W_i\ \indep\ \Bigl(Y_i(0),Y_i(1)\Bigr)\ \Bigl|\ \ow_{g(i)},
\end{equation} 
where $\ow_g$ is the share of treated units in group $g$. This unconfoundedness condition allows us to combine groups with the same value for $\ow_g$. This does not immediately have any empirical content. To see this, consider  the special case with two units per group. This implies  there are three values for $\ow_g$, namely $0$, $1/2$, and $1$.
If  $\ow_g=0$ or $\ow_g=1$, the combined set of groups  has only control units or only treated units, so there is no basis to estimate treatment effects. If we look at a set of groups with $\ow_g=1/2$, comparing treated and control units gives us an average of within-group comparisons of treated and control units based on \ref{eq:sim_exp_unc_clusters_w}. However, within-group comparisons were validated already by (\ref{eq:sim_exp_unc_clusters_old}).

Continuing with this two-units-per-group example, let $M_s$ denote the number of groups in the sample with $\ow_g=s$. Without covariates the fixed effect estimator for $\tau$ in (\ref{een}), which we denote by $\htauf$, reduces to averaging the difference between the treated unit and the control unit over the $M_{1/2}$ groups with exactly one treated and one control unit:
\[ \htauf=\frac{1}{M_{1/2}}\sum_{g:\ow_g=1/2} \hat\tau_g,\hskip1cm {\rm 
where }\ \ 
\hat\tau_g=\sum_{i:g(i) = g}W_i Y_i-\sum_{i:g(i) = g}(1-W_i) Y_i.\]
This fixed effect estimator $\htauf$ does not depend on outcomes for groups with either only treated units $\ow_g=1$, or groups with only control units, $\ow_g=0$.
For groups with $\ow_g=1/2$, $\hat\tau_g$ is the difference between a treated and a control outcome from the same group, so it is a natural, and in fact the only natural, estimator for the average effect within that group given that there are exactly two units from that group in the sample.

We can now characterize what $\htauf$ is estimating in settings with heterogeneous treatment effects under group unconfoundedness. Let $\tau=\mathbb{E}[Y_i(1)-Y_i(0)]$ denote the population average treatment effect and define for $s\in\{0,1/2,1\}$ the weighted average effect
\begin{equation}\label{taus} 
\tau(s)=\mathbb{E}[Y_i(1)-Y_i(0)|\ow_{g(i)}=s].
\end{equation}
This is a critical, though somewhat unusual, object. $\tau(s)$ is a conditional average treatment effect, but the conditioning depends on the sampling scheme and the assignment mechanism. To illustrate the unusual nature of this object,
suppose that for the same population we sampled three units per group instead of two. That would change the values of $s$ for which $\tau(s)$ is defined, and it would potentially change the value of $\tau(0)$ and $\tau(1)$ which are defined under both sampling schemes.

Now consider the overall average effect of the treatment.
This can be expressed in terms of the $\tau(s)$ as
\[ \tau=\pr(\ow_{g(i)}=0)
\tau(0)+
\pr(\ow_{g(i)}=1/2)
\tau(1/2)+
\pr(\ow_{g(i)}=1)
\tau(1).
\]
Under group unconfoundedness (\ref{eq:sim_exp_unc_clusters_w}), $\htauf$ is unbiased for  $\tau(1/2)$. If, as is likely in the heterogenous treatment effect case,  $\tau(s)$ varies by $s$, $\tau(1/2)$ would differ from $\tau$, and thus that $\htauf$ estimates something different from $\tau$. 

To further illustrate this point, suppose there is a group-level variable $\ou_g$ that is related to the propensity score and the outcomes in the following way:
\[
 \mathbb{E}[Y_i(1)-Y_i(0)|\oui]=\oui^2\hskip1cm
 {\rm and}\  \ 
 \pr(W_i=1|L_{g(i)}=L_g)=\ou_g.\]
 Furthermore, suppose $\ou_g$ has the following distribution over the population of groups,
\[
\pr(\ou_g=u)=1/3, \ u\in\{0,1/2,1\}.\]
In this setup, the overall average treatment effect is $\tau=\mathbb{E}[Y_i(1)-Y_i(0)]=E[\oui^2]=5/12$.
However, $\pr(\oui=1/2|\ow_{g(i)}=1/2)=1,$  so that $\tau(1/2)=1/4$. Conditioning on groups with one treated and one control observation shifts the distribution of $\ou_g$ from its marginal  distribution to a distribution that puts more weight on values of $\ou_g$ that make the observation of both a treated and control unit likely. It is easy to see that the interpretation of the fixed effect estimator (\ref{een}) changes  if we change the sampling scheme,  {\it e.g.,} if we sample $N_g=3$ units for each group, or if the assignment mechanism changes.

\subsection{Covariates}\label{sec:simple_cov}

Next, we consider a setting with a single binary covariate $X_i \in \{0,1\}$. 
 Define for each group $g$ and each covariate value $x$ the average treatment effect:
\begin{equation*}
    \tau_l(x)\equiv \mathbb{E}[Y_i(1) - Y_{i}(0)|X_i = x, L_{g(i)} = l].
\end{equation*}
This quantity is well-defined for all $(x,l)$ such that $x$ is in the support of $X_i$ in group $L_g$.

A natural generalization of
the group unconfoundedness assumption in  (\ref{eq:sim_exp_unc_clusters_old}) is the  (conditional) group unconfoundedness assumption: 
\begin{equation}\label{eq:sim_exp_unc_clusters}
W_i\ \indep\ \Bigl(Y_i(0),Y_i(1)\Bigr)\ \Bigl|\ X_i,L_{g(i)}.
\end{equation}
Similarly to (\ref{eq:sim_exp_unc_clusters_w}), it implies a different unconfoundedness condition that does not require within-group comparisons:
\begin{equation}\label{eq:sim_exp_unc1}
W_i\ \indep\ \Bigl(Y_i(0),Y_i(1)\Bigr)\ \Bigl|\ X_i,\ow_{g(i)},\ox_{g(i)},\owx_{g(i)},
\end{equation}
where $\ow_{g(i)},\ox_{g(i)},\owx_{g(i)}$ are group-level averages of the corresponding variables. Condition (\ref{eq:sim_exp_unc1}) justifies combining groups that have identical empirical joint distributions of covariates and treatment indicators. 

Focusing again on the case with two units per group, we can further unpack condition \ref{eq:sim_exp_unc1}. In particular, (\ref{eq:sim_exp_unc1}) allows us to compare units for groups $g$ with variation in the treatment and no variation in the covariate. Formally, let $\uw_g$ be the pair of treatment values $(W_i,W_j)$ for the two units $i$ and $j$ in group $g$, and let $\ux_g$ be the pair of treatment values $(X_i,X_j)$ for the two units $i$ and $j$ in group $g$. We can then define the set 
\[\mathbb{B} \equiv \Bigl\{(\uw_g,\ux_g)\Bigl|\,
\uw_g\in\{(0,1),(1,0)\},\ux_g\in\{(0,0),(1,1)\}
\Bigr\},\]
and consider the average treatment effect 
\begin{equation*}
    \tau_{\mathbb{B}} \equiv 
    \mathbb{E}[\tau_{L_{g(i)}}(X_i)|(\uw_{g(i)},\ux_{g(i)})\in \mathbb{B}].
\end{equation*}
Similarly to $\tau(1/2)$ from the previous section, $\tau_\mathbb{B}$ is generally different from the average treatment effect $\tau$. It  depends both on the sampling scheme and the assignment process. 

The fixed effect specification (\ref{een}), in general, does not consistently estimate $\tau_{\mathbb{B}}$ or any other weighted average effect. To see this, observe that the OLS estimator $\htauf$ for the equation (\ref{een}) is equal to the estimated coefficient  of a  different regression function, namely
\begin{equation}\label{twee} 
Y_i=\alpha+W_i\tau+X_i^\top\beta+\ow_{g(i)}\delta+\ox_{g(i)}^\top\gamma+\eta_i.
\end{equation}
 The numerical  equivalence, first shown in \citet*{mundlak1978pooling}, follows from repeated applications of textbook omitted variable bias formulas.  In general, least squares estimation of regression (\ref{twee}) combines data from all groups, including those with $\ow_g \in \{0,1\}$, making $\htauf$ inconsistent for the average effect.

The alternative version of the fixed effect regression in representation (\ref{twee}) is important because it suggests weakening the conditional group unconfoundedness condition from (\ref{eq:sim_exp_unc1}) by dropping 
$\owx_{g(i)}$, leaving the condition as
\begin{equation}\label{eq:sim_exp_unc}
W_i\ \indep\ \Bigl(Y_i(0),Y_i(1)\Bigr)\ \Bigl|\ X_i,\ox_{g(i)},\ow_{g(i)}.
\end{equation}
This condition encapsulates a key feature of the fixed effect specification that $\ow_{g(i)}$ and $\ox_{g(i)}$ fully capture all relevant differences between groups and that $\owx_{g(i)}$ is not required for this. Restriction (\ref{eq:sim_exp_unc}) implies that we can compare treated and control units as long as they have the same values of $X_i$ and $\osi\equiv (\ow_{g(i)},\ox_{g(i)})$, even though they may belong to groups that have different values of $\owx_{g}$. With two units per group, by construction $\osi$ takes $9$ possible values, but only $3$ of them correspond to groups with variation in treatment status: $\osi \in \left\{\left(\frac12,0\right),\left(\frac12, \frac12\right),\left(\frac12,1\right) \right\}$. For other values of $\osi$ we either have only control or only treated units, and thus condition  (\ref{eq:sim_exp_unc}) is not useful (and neither is (\ref{eq:sim_exp_unc_clusters})).

The three remaining values of $\osi$ lead to conceptually different comparisons. Values $\osi \in \left\{\left(\frac12,0\right),\left(\frac12,1\right) \right\}$ correspond to groups where both units are identical in terms of $X_i$ but vary in their treatment.  For these groups (\ref{eq:sim_exp_unc}) justifies within-group comparisons which were already validated by the more restrictive  assumption
(\ref{eq:sim_exp_unc_clusters}).
This leaves one remaining value, $\osi =\left(\frac12, \frac12\right)$.
If $\osi =\left(\frac12, \frac12\right)$ then by definition the two units in the group have different values of $X_i$, and within-group comparisons have no causal interpretation: the pair of units in the same group are either $(W_i,X_i)=(0,0)$ and $(W_j,X_j)=(1,1)$, or the pair are  $(W_i,X_i)=(1,0)$ and $(W_j,X_j)=(0,1)$. In both cases, the pairs cannot be compared directly. However, (\ref{eq:sim_exp_unc})  justifies combining groups where $W_i = X_i$ with those where $W_i = 1-X_i$. It is useful to consider this comparison in more detail. Suppose that one group has units with $(W_i,X_i)\in\{(0,0),(1,1)\}$, and another  group has units with $(W_i,X_i)\in\{(0,1),(1,0)\}$. In neither of the groups we can directly compare treated and control units. But under (\ref{eq:sim_exp_unc}) we can combine the two groups and then compare a treated unit $i$ with $(W_i,X_i)=(1,1)$ in one group with a control unit $j$ with $(W_j,X_j)=(0,1)$ in another group because they satisfy the two conditions that $(i)$ they have the same value of the covariates ($X_i=X_j=1)$  and $(ii)$ they belong to groups with the same value for $\ow_g$ and $\ox_g$ (namely $\ow_{g(i)}=\ow_{g(j)}=1/2$ and $\ox_{g(i)}=\ox_{g(j)}=1/2$).

Define the conditional expectation of $\tau_{l}(X_i)$ where the expectation is over all units with $X_i=x$ and over all groups with $\os_{g(i)}=s$:
\begin{equation*}
\tau(s,x) \equiv\mathbb{E}[\tau_{L_{g(i)}}(X_i)|\os_{g(i)}=s,X_i=x]=
 \mathbb{E}[Y_{i}(1) - Y_{i}(0)|\os_{g(i)}= s, X_i = x].\\
\end{equation*}
Similar to $\tau(s)$ and $\tau_{\mathbb{B}}$, this object is defined in terms of the sampling scheme and is not a fixed population characteristic. The discussion above shows that we can consistently estimate $\tau(s,x)$ for certain values of $(s,x)$. Importantly, the linear specification (\ref{een}) relies on comparisons that are not validated by  (\ref{eq:sim_exp_unc}) and thus does not estimate a meaningful average effect in this case.

\subsection{Discussion}

Examples from the previous sections show that we can reduce group unconfoundedness, as in  (\ref{eq:sim_exp_unc_clusters_old}) or (\ref{eq:sim_exp_unc_clusters}), to conditioning on some balancing score $\osi$ and covariates, as in (\ref{eq:sim_exp_unc_clusters_w}) or (\ref{eq:sim_exp_unc1}).  Despite the lack of immediate empirical content, this reduction brings two critical  insights.
Both insights are about operationalizing the notion that some groups are more similar than others, which is absent in
(\ref{eq:sim_exp_unc_clusters_old}) or (\ref{eq:sim_exp_unc_clusters}).

First, conditioning on an unordered discrete characteristic -- group indicator -- does not provide a distance metric to make the case that  group $g$ is closer to group $g'$ or to group $g''$. In contrast, (\ref{eq:sim_exp_unc_clusters_w}) (and similarly (\ref{eq:sim_exp_unc1})) allows for the establishment of such a metric: a group $g$ with $\ow_g=0.10$ is closer to group $g'$ with $\ow_{g'}=0.11$ than it is to group $g''$ with $\ow_{g''}=0.50$ (see \citet*{johannemann2019sufficient} for an application of this idea). The second  insight is even more central to the current paper. It allows us to find a middle ground between fully conditioning on the entire set of group indicators and not conditioning on these indicators at all. For example, going from (\ref{eq:sim_exp_unc1}) to (\ref{eq:sim_exp_unc}) justifies some, but not all between-group comparisons. Depending on the application, one can go further and assume a stronger condition
   \begin{equation*}
W_i\ \indep\ \Bigl(Y_i(0),Y_i(1)\Bigr)\ \Bigl|\ X_i,\ow_{g(i)}. 
\end{equation*}
Alternatively, if we observe more units in each group, we can relax (\ref{eq:sim_exp_unc}) to improve the comparability of groups without reducing it fully to (\ref{eq:sim_exp_unc_clusters}).
We discuss in Section \ref{section:identification} whether and when it is enough to adjust for some subset of these averages.

Another issue illustrated in the previous sections is the relationship between the unconfoundedness conditions and the fixed effects specification (\ref{een}). In general, this regression does not estimate the average treatment effect, and in the presence of covariates, it might fail to estimate even a weighted average effect. However, under the restriction (\ref{eq:sim_exp_unc}) we can go beyond the linear specifications (\ref{een}) and consider different estimation methods. First, we can use more flexible specifications for the regression function as a function of the control variables $(X_i,\ow_{g(i)},\ox_{g(i)})$ and the treatment $W_i$:
\[ Y_i(w)=\mu(w,X_i,\ow_{g(i)},\ox_{g(i)})+\varepsilon_i(w),\]
with a parametric or non-parametric specification for $\mu(w,\cdot)$ that generalizes the linear additive form in (\ref{een}).
These specifications may include higher-order moments of the control variables, interactions with the treatment, or transformations of the linear index.
Given estimates of $\mu(w,\cdot)$ we can average the difference $\hat \mu(1,X_i,\ow_{g(i)},\ox_{g(i)})-\hat \mu(0,X_i,\ow_{g(i)},\ox_{g(i)})$ over the sample to estimate the average treatment effect.

Second, we can model the propensity score as a function of $(X_i,\ow_{g(i)},\ox_{g(i)})$:
\[ e(x,\overline w,\overline x)\equiv {\rm pr}(W_i=1|X_i=x,\ow_{g(i)}=\overline w,\ox_{g(i)}=\overline x).\]
Once we have estimates of the propensity score, we can use them to develop inverse propensity score weighting estimators (see \cite{arpino2011specification} for an application of this idea to grouped data). In particular, an attractive approach would be to use the inverse propensity score weighting or other forms of balancing in combination with a credible specification of the regression function. For example, one could use a weighted version of (\ref{twee}) to make the results more robust to misspecification of the regression function. Such doubly-robust and balancing methods have been argued in the recent causal inference literature on unconfoundedness to be more robust to misspecification than estimators that rely solely on specifying the conditional mean of the outcome given conditioning variables and treatments ({\it e.g., } \citet*{robins1995semiparametric,  zubizarreta2015stable,
chernozhukov2017double, athey2018approximate}).

Following this discussion, we recommend using three modifications of (\ref{een}). First, choose what group characteristics one wishes to include in the analysis beyond the averages of the treatment and the covariates. Second, specify a flexible conditional mean function that allows for dependence on the selected group characteristics. Third,  estimate the propensity score as a function of individuals and group characteristics and combine that with the conditional mean specification to obtain a more robust estimator for the average treatment effect.

\subsection{Related Literature}

Our approach is related to classic panel data literature, \textit{e.g.}, \cite{mundlak1978pooling, chamberlain1984panel,chamberlain2010binary,altonji2005cross}. Similar to these papers, we develop a conditional approach, where we effectively control for heterogeneity across groups using group-level balancing scores. Our key contribution to this literature is to show how these statistics arise naturally from the design model and connect this approach to modern estimation algorithms developed for cross-sectional data.

Our results are also connected to random effects models studied in the more recent panel data literature (see \citet{arellano2011nonlinear} for a recent review). For example, \citet{arellano2016nonlinear} explicitly model the distribution of the outcomes and assignments given the individual unobservable characteristics and achieve identification using results from nonparametric nonlinear deconvolution literature (\citet{hu2008instrumental}). The key conceptual and practical difference between our approach and theirs is that we do not use the outcome data to deal with unobserved heterogeneity and rely only on the distribution of treatments and covariates for identification.  In this sense, we follow a tradition used in causal inference literature (\citet*{rosenbaum_book, imbens2015causal, abadie2020sampling}).

Our approach is connected to grouped fixed effects strategies (\citet*{hahn2010panel, bonhomme2015grouped,bonhomme2022discretizing}). Similar to this literature, we propose pooling data across groups with similar aggregate statistics. There are two main differences, though. As discussed above, we do not use the outcome data, instead focusing on the assignment process. Also, our identification and inference results are valid for groups of small size, whereas statistical results in the other literature are based on approximations with large groups.  

Finally, our results are related to recent literature on regression models with fixed effects (\textit{e.g.},\citet{de2020two,goodman2021difference,wooldridge2021two}). Similar to these papers, we discuss the interpretation of the least-squares estimators under heterogeneity in treatment effects and develop an estimator that is robust to the presence of such heterogeneity. Our focus, however, is quite different: instead of restricting the conditional outcome distributions, we focus on the assignment model and achieve robustness by using it to construct the estimator. 

\section{Identification, Estimation, and Inference}
\label{section:identification}

In this section we present the main results.

\subsection{Setup}

We start with the sampling assumption that describes the relationship between groups and unit-level data:
\begin{assumption}\label{as:random_sampling}{\sc(Grouped sampling)}
We randomly sample $\numcl$ groups out of a super-population of groups with $g\in\{1,\dots, \numcl\}$ being a generic group; we then randomly sample $N_g$ units from each group $g$, so that $N=\sum_{g=1}^\numcl N_g$ is the total sample size. $N_g\ge2$ and is the same for all $g=1,\ldots,\numcl.$
\end{assumption}
The assumption that the group size is identical for all groups is for expositional reasons only. In the Appendix, we generalize our results  to allow for heterogeneity in the group sizes. 
The key complication is that it requires us to add the group size to the set of conditioning variables that make up the balancing score.
We assume that for a generic group $g$ we observe its label $L_g$. For each unit $i$ define $g(i)\in\{1,\ldots,M\}$ -- the index of the group this unit belongs to. For each unit $i$ we observe the outcomes $Y_i$, binary treatment indicators $W_i\in\{0,1\}$, and covariates $X_i \in \mathbb{X}$. Assumption \ref{as:random_sampling} allows us to view the data as $\numcl$ independent and identically distributed copies of random elements $\left(L_g,\{(W_i,X_i,Y_i)\}_{i: g(i)= g}\right)$, with the additional independence among vectors $\{(W_i,X_i,Y_i)\}_{i: g(i)= g}$.

We interpret the observed outcomes using Neyman-Rubin potential outcomes model (\cite{neyman1923,rubin1977assignment,imbens2015causal}):
\begin{equation}\label{eq:sutva}
Y_i = Y_i(1)W_i + Y_i(0)(1-W_i),
\end{equation}
which implies the absence of spillovers across units and groups. We impose restrictions on the treatment assignment process:
\begin{assumption}\label{as:unconf_cl_un}{\sc (Group-level Unconfoundedness)}\\
\begin{equation*}
W_i\ \indep\ \Bigl(Y_i(0),Y_i(1)\Bigr)  \ \Bigl|\ X_i,L_{g(i)}.
\end{equation*}
\end{assumption}
\noindent This assumption restricts the distribution of potential outcomes and treatment indicators within a particular group. It implies that we can give a causal interpretation to comparisons of units with the same characteristics within the group.

\subsection{Identification Results}

For the first identification result we need to introduce additional notation. For each group $g$ define $\mathbb{P}_g$  to be the empirical measure of $(W_i,X_i)$ in  group $g$. Formally, for any set $A\subseteq \mathbb{X}\times \{0,1\}$, $\mathbb{P}_g(A)$ is defined through the following relationship:
\begin{equation*}
    \mathbb{P}_g(A) =\frac{1}{N_g} \sum_{i:g(i)= g} \mathbf{1}_{(W_i, X_i) \in A}.
\end{equation*}
We use this construction for our first identification result:
\begin{prop}\label{prop:unc_in_general}{\sc(Unconfoundedness with empirical measure)}
Suppose  Assumptions  \ref{as:random_sampling} and \ref{as:unconf_cl_un} hold. Then:
\begin{equation*}
W_i \indep\   \Bigl(Y_i(0),Y_i(1)\Bigr)  \ \Bigl|\ X_i,\mathbb{P}_{g(i)}
\end{equation*}
\end{prop} 
\noindent For proofs of all results in this section see Appendix \ref{app:ident}. Proposition \ref{prop:unc_in_general} states that as long as units have the same characteristics, and they come from groups identical in terms of $\mathbb{P}_g$ they are comparable. To relate this result to the discussions in the previous section, assume that there are no covariates. In this case, $\mathbb{P}_g$ reduces to $\ow_g$ and we get condition (\ref{eq:sim_exp_unc_clusters_w}). If we introduce a single binary covariate -- we get condition (\ref{eq:sim_exp_unc1}). This result is a manifestation of a familiar  balancing property of the propensity score  (\textit{e.g.,}  \citet{rosenbaum1983central}). The subpopulation with the same value for  $(X_i,\mathbb{P}_g)$ are balanced: the distribution of treatments is the same for all units within such subpopulations. 

However, as discussed in the previous section, the empirical relevance of this result is limited. To make it operational, we need to substitute the empirical distribution of $(W_i, X_i)$ with a low-dimensional object. To do this, we put structure on the joint distribution of $(W_i,X_i)$ and its relation to $L_g$ with an exponential family representation (\citet{cox1979theoretical}).
\begin{assumption}\label{as:exp_family}{\sc (Exponential family)}
Conditional on $L_g=l$ distribution of  $(W_i,X_i)$ belongs to an exponential family with a known sufficient statistic, {\it i.e.,} its conditional density 
has the following form:
\begin{equation}\label{eq:suf_stat}
f_{W_i,X_i|L_g}(w,x|l)= h(w,x) \exp\Bigl\{ \eta(l)^\top S(w,x)+\eta_0(l))\Bigr\},
\end{equation}
with potentially unknown carrier $h(\cdot)$.
\end{assumption}
Under relatively weak conditions distributions can be well approximated by distributions in exponential families (\citet{barron1991approximation}), so
as long as  we do not restrict the dimension of $S(\cdot)$ this  assumption is fairly weak.  For example, note that if the distribution of $X_i$ is discrete, one can immediately write the joint distribution of $(W_i,X_i)$ within each group as an exponential family distribution with a group-specific parameter. In addition, we can approximate any distribution arbitrarily well by a discrete distribution.  Of course, it is only when the dimension of $S(\cdot)$ is small relative to the number of observations that the assumption will be seen to be useful.

Assumption \ref{as:exp_family} places no restrictions on the distribution of the potential outcomes and thus does not restrict heterogeneity in treatment effects. At the same time, it restricts the joint distribution of $(W_i, X_i)$, not just the conditional distribution of $W_i$. This is without loss of generality: if $X_i$ is discrete and function $S(W_i, X_i)$ includes indicators for all potential values of $X_i$, then (\ref{eq:suf_stat}) places no restrictions on the marginal distribution of $X_i$ given $L_{g(i)}$. For continuous $X_i$ the same argument holds using an infinite set of indicators.

With Assumption (\ref{as:exp_family}) we can define an unobserved group-level variable that captures all the information about the distribution of $(W_i,X_i)$:
\begin{equation*}
    \ou_g\equiv \eta(L_g).
\end{equation*}
We use the bar notation to indicate that this is a group-level variable.
Our next lemma shows that we can substitute $L_g$ with $\ou_g$ for the purpose of identification.
\begin{lemma}\label{lemma:re}
Suppose Assumptions \ref{as:random_sampling}- \ref{as:exp_family} hold. Then: $(i)$ 
\begin{equation*}
\Bigl\{(W_i,X_i)\Bigr\}_{i: g(i)= g} \indep\ L_g\ \Bigr|\ \ou_g
\end{equation*}
and $(ii)$:
\begin{equation*}
W_i\ \indep\ \Bigl(Y_i(0),Y_i(1)\Bigr)\ \Bigl|\ X_i,\oui.
\end{equation*}
\end{lemma}
\noindent This result demonstrates that our problem can be converted into a traditional random effects setup (\textit{e.g.}, \citet*{chamberlain1984panel}): the group labels $L_g$ are not important, only the group-level characteristics $\ou_g$, which are functions of these labels, are.
We do not know the group-specific parameters $\overline{U}_g$, nor can we estimate them consistently in settings with a small number of observations per group because of the incidental parameter problem (\citet*{neyman1948consistent, lancaster2000incidental}).
However, and this is a key insight in the paper, it is {\it not} necessary to 
estimate them consistently because the sufficient statistic in the exponential family representation
in \ref{as:exp_family},
$\os_g\equiv\sum_{i:g(i) = g}S(W_i,X_i)/N_g$, plays the role of a balancing statistic. 

Using Assumption \ref{as:exp_family} we state our main identification result:
\begin{theorem}\label{th:unc_with_exp_family}
Suppose Assumptions \ref{as:random_sampling}-\ref{as:exp_family} hold. Then:
\begin{equation*}
W_i \indep\   \Bigl(Y_i(0),Y_i(1)\Bigr) \ \Bigl|\ X_i, \osi.
\end{equation*}
\end{theorem}

\noindent Theorem \ref{th:unc_with_exp_family} is related to Proposition \ref{prop:unc_in_general}, but the key difference is that Theorem \ref{th:unc_with_exp_family} has empirical content as long as $\osi$ is less general than $\mathbb{P}_{g(i)}$.  Theorem \ref{th:unc_with_exp_family} reduces the high-dimensional object $\mathbb{P}_{g(i)}$ to an average $\osi$ and allows us to combine multiple groups. This makes cross-sectional methods developed in the last decades feasible for applications with grouped data. In practice, the importance of this point depends on the number of units per group and the dimension of $S(\cdot)$. The higher the dimension of $S(\cdot)$, the closer we are to controlling for $\mathbb{P}_{g(i)}$ and thus relying solely on comparisons within groups. 

Theorem \ref{th:unc_with_exp_family} becomes useful once we fix a particular low-dimensional $S(W_i,X_i)$. We assume that $S(\cdot)$ is known throughout most of the paper. This makes Assumption \ref{as:exp_family} restrictive and we can no longer argue that it holds for an arbitrary conditional distribution of $(W_i,X_i)$. As we illustrate with the examples in the previous section, this is needed if we want to go beyond within-groups analysis without imposing additional assumptions on the potential outcomes. In the current empirical practice, researchers rely on (\ref{een}) or, equivalently, (\ref{twee}), thus implicitly using $S(W_i,X_i) = (W_i,X_i)$, which can be too coarse for some applications. In particular, this choice assumes that covariates do not interact with group labels in the assignment model. We informally discuss selecting the set of balancing scores in Section \ref{sec:discussion}.

Define propensity score:
\begin{equation*}
e(x,s) \equiv{\rm pr}(W_i =1|X_i=x, \osi=s).
\end{equation*}
Depending on the value of the balancing score, the propensity score  $e(x,s)$ might be equal to zero or one for certain values of $s$. For example, this happens if $\ow_g$ is part of the balancing score. As a result, we cannot expect the overlap assumption (\citet{imbens2015causal}) to hold. Instead, we are making the following assumption:
\begin{assumption}\label{as:overlap}{\sc(Known overlap)}
For a known set $\mathbb{A}$ with ${\rm pr}\left((X_i,\osi) \in \mathbb{A}\right)>0$ there exists $\eta >0$, such that $(i)$ for any $(x,s) \in  \mathbb{A}$ we have $\eta <e(x,s)<1-\eta$, and $(ii)$  for $(x,s) \not\in  \mathbb{A}$ $e(x,s) \in \{0,1\}$.
\end{assumption}
This assumption has two parts: the first part restricts $e(x,s)$ to be non-degenerate on a certain set and only that set. This is necessary if we want to identify treatment effects without relying on functional form assumptions. The second part is different: we assume that the set is known to a researcher. This is a generalization of the standard overlap assumption, where we assume that the set $\mathbb{A}$ is equal to the support of the covariate space, see \citet*{crump2009dealing}. 

To provide additional intuition for Assumption \ref{as:overlap}, let us go back to the example from Section \ref{sec:simple_cov}, where $S(W_i,X_i) = (W_i,X_i)$ and we observe two units per cluster. In that case, the propensity score is equal to either zero or one for any $\osi \not\in \left\{\left(\frac12,0\right),\left(\frac12, \frac12\right),\left(\frac12,1\right) \right\}$. By construction $e\left(0, \left(\frac12,0\right)\right) = e\left(1, \left(\frac12,1\right)\right) = \frac12$, but the values of $e\left(0, \left(\frac12, \frac12\right)\right)$ and $e\left(1, \left(\frac12, \frac12\right)\right)$ are not fixed, except for the restriction $e\left(0, \left(\frac12, \frac12\right)\right) + e\left(1, \left(\frac12, \frac12\right)\right) = 1$. As a result, if we set $\mathbb{A}\equiv \left\{(x,s):\left(0,\left(\frac12,0\right)\right),\left(0,\left(\frac12, \frac12\right)\right),\left(1,\left(\frac12, \frac12\right)\right),\left(1,\left(\frac12,1\right)\right) \right\}$, then this amounts to assuming $\eta <e\left(0, \left(\frac12, \frac12\right)\right) < 1-\eta$. 

We can now define our target estimands. We focus on two of them; the first one is similar to the objects considered in the previous section:
\begin{equation}\label{eq:pop_est}
    \tau_{\mathbb{A}}\equiv \mathbb{E}[(Y_{i}(1) - Y_{i}(0))| (X_i, \osi) \in \mathbb{A} ].
\end{equation}
The second object is a sample group-weighted version of (\ref{eq:pop_est}):
\begin{equation*}
    \tilde\tau_{\mathbb{A}}\equiv \frac{\sum_{i=1}^N\mathbf{1}_{(X_i,\osi) \in \mathbb{A}}\mathbb{E}[Y_{i}(1) - Y_{i}(0))|X_i,\osi
    ]}{\sum_{i=1}^N\mathbf{1}_{(X_i,\osi) \in \mathbb{A}}}. 
\end{equation*}
Theorem (\ref{th:unc_with_exp_family}) together with Assumption (\ref{as:overlap}) implies that $\tau_{\mathbb{A}}$ and $\tilde\tau_{\mathbb{A}}$ are identified. We formally state this result in the next Corollary.
\begin{corollary}\label{cor:identification}
Suppose Assumptions
\ref{as:random_sampling}--\ref{as:overlap} hold. Then $\tau_{\mathbb{A}}$ is identified.
\end{corollary}
\begin{proof}
The results follow from the identity:
\begin{multline*}
    \mathbb{E}[(Y_{i}(1) - Y_{i}(0))| (X_i, \osi) \in \mathbb{A} ] =\\ \mathbb{E}\left[\left(\frac{W_i}{e(X_i, \osi)} - \frac{1-W_i}{1-e(X_i, \osi)}\right)Y_i \Bigl|(X_i, \osi) \in \mathbb{A} \right],
\end{multline*}
and the fact that $e(X_i,\osi)$ being a conditional expectation of a binary variable is identified under Assumption \ref{as:random_sampling}.
\end{proof}

In the remainder of the paper, we primarily focus on the properties of estimators for $\tilde\tau_{\mathbb{A}}$. 
Such estimators can also be used to estimate $\tau_{\mathbb{A}}$, but the variances are different.
Our focus on $\tilde\tau_{\mathbb{A}}$ is analogous to the focus on the sample average treatment effect in the experimental literature.

\subsection{Discussion}\label{sec:discussion}

Our key  result --  Theorem \ref{th:unc_with_exp_family} -- is related to some well-known identification results. For example, in \citet*{altonji2005cross} a key assumption (Assumption 2.1) requires  that there is an observed variable $Z_i$ such that conditioning on $Z_i$ renders the covariate of interest (the treatment in our case) exogenous. In our setting, the role of this conditioning variable is  played by the balancing score $ \osi$. We show how this property can arise from assumptions on the joint distribution of the  treatment and the other covariates, and how we can make this more plausible by increasing the dimension of balancing scores.

Assumption \ref{as:exp_family} plays the key role in Theorem \ref{th:unc_with_exp_family} and in the previous section we argued that it can be viewed as a natural approximation for the conditional distribution of covariates and treatment indicators. In addition to this statistical justification, it arises in a structural economic model. To this end, consider a setup where each unit is characterized by $(W_i, X_i)$ and then chooses group membership $l$ out of set $L$. For example, $W_i$ can be a scholarship or a voucher, of a student with observed characteristics $X_i$, and $L_{g(i)}$ can be a college or a school this individual has chosen. In this case, $\mathbb{E}[Y_{i}(1) - Y_{i}(0)|L_{g(i)}]$ is an average direct effect of treatment unmediated by the choice $L_{g(i)}$.

Suppose the following function is the indirect utility for individual $i$ from group $l$:
\begin{equation*}
    V_i(l) = V(W_i, X_i, l) + \epsilon_{i}(l)
\end{equation*}
where $\epsilon_{i}(l)$ are i.i.d. extreme-value type-1 shocks. In this case, the probability of selecting a particular option $g$ has the following form:
\begin{equation*}
    \mathbb{E}[\mathbf{1}_{L_{g(i)} = l}|W_i,X_i] = \frac{\exp(V(W_i, X_i, l))}{\sum_{k\in L}\exp(V(W_i, X_i, k)) }.
\end{equation*}
This restricts the joint distribution of $(L_{g(i)},W_i, X_i)$:
\begin{equation*}
f_{W_i,X_i,L_{g(i)}}(w,x,u) =  f(w,x)\exp\left(V(w,x,l)\right),
\end{equation*}
and thus as long as $V(w,x, u) =\eta^\top(l)S(w,x)+\eta_0(l)$ Assumption \ref{as:exp_family} is satisfied. This discussion shows that exponential family restriction is a natural approximation in a structural model of choice.

In practice, balancing scores are unknown and need to be selected by researchers. This choice should reflect domain knowledge and a full treatment of this problem is an open question. However, we provide a suggestion for systematically selecting balancing scores in the case where we have a large set of potential balancing scores that includes all the relevant ones but also some that are not relevant. Intuitively we would like a selection procedure to select more balancing scores in settings where we have a lot of units per group and if the distributions vary substantially across groups. Assumption \ref{as:exp_family} implies that the conditional probability that a unit in the sample is from group $l$, conditional on $(W_i,X_i)$ and conditional on the set of $L_1,\ldots,L_M$, has a multinomial logit form:
\[ {\rm pr}(L_{g(i)}=l|W_i,X_i,L_1,\cdots,L_M)=\frac{\exp(\eta_0(L_g)+\eta(L_g)^\top S(W_i,X_i))}
{\sum_{k=1}^\numcl \exp(\eta_0(L_{k})+\eta(L_{k})^\top S(W_i,X_i))}.\]
Hence the problem of selecting the  balancing scores is similar to the problem of selecting covariates in a multinomial logistic regression model. Given a large set of potential  balancing scores, we can use standard covariate selection methods, such as LASSO (\citet*{tibshirani1996regression}) to select a sparse set of relevant ones. In practice,   balancing scores are likely to be correlated, and other selection methods such as Adaptive Elastic Net (\citet*{zou2009adaptive}) might be more appropriate. See also \citet*{de2019identifying} for the application of this algorithm to the identification of unobserved network structure with panel data.

\subsection{Estimation and Inference}
This section collects several inference results for a class of semiparametric estimators of $\tau_{\mathbb{A}}$ and $\tilde\tau_{\mathbb{A}}$. All proofs can be found in Appendix \ref{ap:sem_inf}. For further use, we use the following notation for the conditional mean, propensity score, and residuals:
\begin{equation*}
\begin{cases}
\mu(W_i,X_i,\osi)  \equiv \mathbb{E}[Y_i|W_i,X_i,\osi],\\
e(X_i,\osi) \equiv  {\rm pr}(W_i=1|X_i,\osi),\\
\varepsilon_i(w) \equiv  Y_i(w) - \mu(w,X_i,\osi).
\end{cases}
\end{equation*}
These expectations are defined using Assumption \ref{as:random_sampling}, which determines the distribution of $\osi$. In particular, as we discussed in Section \ref{section:simple}, the distribution of $\osi$ changes with the number of observed units in each group. We restrict moments of the corresponding errors:
\begin{assumption}\label{as: moment_cond}{\sc(Moment conditions)} For $w \in \{0,1\}$
\begin{equation*}
 \mathbb{E}[\varepsilon_i^2(w)|\obi] < K \text{ a.s.},\quad  \mathbb{E}[\varepsilon_i^4(w)]<\infty.
\end{equation*}
\end{assumption}
\noindent These conditions hold for bounded potential outcomes, which is the leading case in practice. In the case of unbounded outcomes, the first condition imposes a restriction on the degree of heteroscedasticity.

We will use $\hat \mu_g(\cdot)$ and $\hat e_g(\cdot)$ for generic estimators of $\mu(\cdot)$ and $e(\cdot)$. Subscript $g$ is used to allow for cross-fitting (\textit{e.g.,} \citet{chernozhukov2018double}); we assume that researchers use $K$-fold group-level cross-fitting. In particular, researchers split the observed $M$ groups into $K$ random subsamples, and for each group $g$ construct estimators $\hat \mu_g(\cdot)$, $\hat e_g(\cdot)$ using the data from $K-1$ subsamples that do not contain group $g$. Define indicator variable $A_i \equiv \mathbf{1}_{(X_i,\osi) \in \mathbb{A}}$ and let $\ova\equiv\frac{1}{\numcl} \sum_{g=1}^M\frac{1}{N_g}\sum_{i:g(i) = g}A_i$ be the estimate of the share of units for which we have overlap. We assume  the generic estimators $\hat e_g$ and $\hat \mu_g$ satisfy several high-level consistency properties standard in the program evaluation literature: 

\begin{assumption}\label{as:cons_prop_score}{\sc (High-level conditions)} The following conditions are satisfied for $\hat e_g(\cdot)$ and $\hat \mu_g(\cdot)$ as $M$ goes to infinity: $(i)$  $\frac{\eta}{2} < \hat e_g(X_i,\osi) <1-\frac{\eta}{2}$ on $\mathbb{A}$, $(ii)$ mean-consistency, where $w \in \{0,1\}$:
  \begin{equation*}
 \begin{aligned}
  &\mathbb{E}\left[\frac{1}{\numcl}\sum_{g=1}^{\numcl} \frac{1}{N_g}\sum_{i: g(i)= g} A_i(e(X_i,\osi) - \hat e_g(\obi))^2\right]  = o(1),\\
  &\mathbb{E}\left[\frac{1}{\numcl}\sum_{g=1}^{\numcl} \frac{1}{N_g}\sum_{i: g(i)= g}  A_i(\mu(w,X_i,\osi) - \hat \mu_g(w,X_i,\osi))^2\right]  = o(1),\\
 \end{aligned}
 \end{equation*}
 $(iii)$ and product rate condition:
 \begin{multline*}
     \frac{1}{\numcl}\sum_{g=1}^\numcl\frac{1}{N_g}\sum_{i: g(i)= g}A_i \left(\hat  \mu_g(W_i,X_i,\osi) -\mu(W_i,X_i,\osi)\right)^2\\
\times\frac{1}{\numcl}\sum_{g=1}^\numcl\frac{1}{N_g}\sum_{i: g(i)= g}A_i\left(\hat e_g(X_i,\osi)- e(X_i,\osi) \right)^2 =
   o_p\left(\frac{1}{\numcl}\right).
 \end{multline*}
\end{assumption}
Under Assumptions \ref{as:random_sampling}, \ref{as:overlap}, and bounded outcomes these conditions are satisfied if covariates are discrete. To see this note that $(i)$ and $(ii)$ are weak consistency restrictions, while each term in $(iii)$ is $O_p\left(\frac{1}{M}\right)$. Under appropriate smoothness restrictions (see \citet*{hansen2022econometrics} for a textbook treatment) these conditions hold for conventional kernel estimators. The use of cross-fitting expands the set of available estimators, allowing researchers to use machine learning estimators (see \citet*{newey2018cross} for details).

For arbitrary functions $m(\cdot), p(\cdot)$ define the following functional (\textit{e.g.}, \citet{robins1995semiparametric,
hahn1998role,
chernozhukov2018double}):
\begin{multline*}
\psi(y,w,x,s,m(\cdot),p(\cdot)) \equiv  m(1,x,s) -m(0,x,s) + \\
\left(\frac{w}{p(x,s)} - \frac{1-w}{1-p (x,s)}\right)(y-m(w,x,s),
\end{multline*}
and its group-level aggregate:
\begin{equation*}
     \rho_g(m(\cdot),p(\cdot))\equiv \frac{1}{N_g}\sum_{i:g(i) = g} A_i
\psi(Y_i,W_i,X_i,\osi,m(W_i,X_i,\os_i),p(X_i,\os_i)).
\end{equation*}
We use $\rho_g(\cdot)$, estimators $\hat \mu_g(\cdot), \hat e_g(\cdot)$, and $\ova$ to define our proposed Generalized Mundlak estimator (GME):
\begin{equation}\label{eq:dr_est}
    \hat \tau_{\rm GME} \equiv  \frac{\sum_{g=1}^{\numcl} \rho_g(\hat \mu_g,\hat e_g)}{\ova}.
\end{equation}
We also define the influence function:
\begin{equation*}
    \xi_g\equiv \sum_{i: g(i)= g}\frac{A_i}{N_g} \left(\frac{W_i}{e(\obi)}-\frac{1-W_i}{1-e(\obi)}\right)(Y_i - \mu(W_i,\obi)). 
\end{equation*}
Our next result shows the properties of $\hat \tau_{\rm GME}$:
\begin{theorem}{\sc(Inference)}\label{th:inference}
Suppose Assumptions  \ref{as:random_sampling}--\ref{as: moment_cond} hold. Then as $M$ goes to infinity $(i)$:
\begin{equation}\label{eq:est_def}
\hat\tau_{\rm GME}= \tilde\tau_\mma+o_p(1) = \tau_\mma + o_p(1),
\end{equation}
and $(ii)$:
\[ \sqrt \numcl(\hat\tau_{\rm GME}-\tilde \tau_{\mma})\stackrel{d}{\longrightarrow}{\cal N}(0,\mmv),\hskip1cm {\rm where}\ \ 
 \mmv=\frac{\mme\left[\xi_g^2\right]}{\mathbb{E}^2[A_i]}.\]
\end{theorem}
This result allows us to conduct inference for $\tilde \tau_\mma$ -- the sample group-weighted version of our estimand. A similar result holds for $\tau_\mma$ with a larger variance that reflects additional randomness in $\tilde\tau_\mma$. To estimate the asymptotic variance $\mmv$ we define the estimated version of $\xi_g$:
\begin{equation*}
\hat\xi_g\equiv \sum_{i: g(i)= g}\frac{A_i}{N_g} \left( \left(\frac{W_i}{\hat e_g(\obi)}-\frac{1-W_i}{1-\hat e_g(\obi)}\right)(Y_i - \hat \mu_g(W_i,\obi)) \right),
\end{equation*}
and use it to construct the estimator:
\begin{equation}\label{eq:dr_var}
\hat \mmv: =\frac{1}{\ova^2}
\frac{1}{\numcl}\sum_{g=1}^\numcl \left(\hat\xi_g-\frac{1}{\numcl}\sum_{g'=1}^\numcl \hat\xi_{g'}\right)^2.
\end{equation}
Next result shows that this estimator is consistent:
\begin{theorem}{\sc(Variance consistency)}\label{th:var_cons}
Suppose Assumptions of Theorem \ref{th:inference} hold. Then the variance estimator is consistent:
\begin{equation*}
\hat \mmv=\mmv +o_p(1).
\end{equation*}
\end{theorem}
Together Theorems \ref{th:inference}, \ref{th:var_cons} imply that standard confidence intervals are asymptotically valid. In particular, if $z_{\alpha}$ is the $\alpha$-quantile of the standard normal distribution then asymptotically
\begin{equation*}
    \tilde\tau_\mma \in \hat \tau_{\rm GME} \pm \sqrt{\frac{\hat \mmv}{\numcl}}z_{\alpha/2}
\end{equation*}
with probability $(1-\alpha)$.

\section{Extensions}

In this section, we discuss three extensions of the ideas introduced in this paper.

\subsection{Quantile Treatment Effects}
Theorem \ref{th:unc_with_exp_family} states that conditional on the covariates and the  balancing scores we have the unconfoundedness condition:
\[W_i \indep\   \Bigl(Y_i(0),Y_i(1)\Bigr)  \ \Bigl|\ X_i, \osi.\]
This implies that we can study estimation of effects other than average treatment effects. This is important in applications where we want to estimate, distributional effects controlling for group-level unobserved heterogeneity.

In particular, for any bounded function $f: \mathbb{R}\rightarrow \mathbb{R}$ we can estimate $\mathbb{E}[f(Y_i(w))]$ using the inverse propensity score:
\begin{equation*}
\mathbb{E}[f(Y_i(w)] = \mathbb{E}\left[\frac{\mathbf{1}_{W_i = w}f(Y_{i})}{e(X_i,\osi)}\right]
\end{equation*}
This allows us to identify quantile treatment effects of the type introduced by \citet*{lehmann2006nonparametrics}. If we are interested in $q$-th quantile of the distribution of $Y_i(w)$ then (under appropriate continuity) we can identify it as a solution to the following problem:
\begin{equation*}
c(q): \mathbb{E}\left[\frac{\mathbf{1}_{W_i = w}\mathbf{1}_{Y_i \le c(q)}}{e(X_i,\osi)}\right]  = q
\end{equation*}
For the standard case under unconfoundedness \citet*{firpo2007efficient} has developed effective estimation methods that can be adapted to this case.

\subsection{Beyond Exponential Families}
Modeling the conditional distribution of $(W_i,X_i)$ given $\ou_g$ using an exponential family is natural for the purposes of this paper. Nevertheless, in some applications, other families can be more appropriate.  In particular, another operational choice is a discrete mixture. Assume that $\oui$ can take a finite number of values $\{u_1,\dots, u_p\}$ with probabilities $\pi_{1},\dots, \pi_{p}$ and the conditional distribution of $W_i, X_i$ given $\oui$ is given by $f(w,x|u)$. Collect all the data that we observe for group $g$ in the following tuple:
\begin{equation*}
\mathcal{D}_g \equiv  ((X_{1},W_{1}),\dots, (X_{N_g}, W_{N_g})).
\end{equation*}
The marginal distribution of this object is given by the following expression:
\begin{equation*}
f_{\mathcal{D}_g}(x_1,w_1,\dots, x_{N_g}, w_{N_g}) =\sum_{k=1}^p\prod_{j =1}^{N_g} f(x_j, w_j| k)\pi_k
\end{equation*}
This implies that the conditional distribution of $\ou_g$ given $\mathcal{D}_g$ has the following form:
\begin{equation*}
\pi(\ou_g = k|x_1,w_1,\dots, x_{N_g}, w_{N_g}) = \frac{\prod_{j =1}^{N_g} f(x_j, w_j| k)\pi_k }{\sum_{k=1}^p\prod_{j =1}^{N_g} f(x_j, w_j| k)\pi_k}.
\end{equation*}
Define $\os_g\equiv  (\pi(\ou_g = 1|\mathcal{D}_g), \dots, \pi(\ou_g = p|\mathcal{D}_g))$ and observe that as long as  assignment is unconfounded given $(X_i,\oui)$ we have the following:
\[ W_i\  \indep\  \Bigl(Y_i(0),Y_i(1)\Bigr)\ \Bigl|\  X_i,\osi.\]
Recent results (\textit{e.g.}, \citet*{allman2009identifiability}) show that $(\pi_1,\dots, \pi_p)$ and $f(w,x|u)$ are nonparametrically identified under quite general assumptions, and \citet*{bonhomme2016estimating} provides a way of estimating these objects. Using these methods we can construct $\hat\os_g$ and use it as a  balancing score. 

\subsection{Beyond Binary Treatments}

So far, we have discussed only the case of binary treatments. Applications with non-binary treatments are also common in empirical work, and our strategy has a natural extension for this case. Following most of the applied work, we consider a linear model for the potential outcomes:
\begin{equation*}
    Y_i(w) = \alpha_i + \tau_i w,
\end{equation*}
where $w$ does not have to be binary. Here heterogeneity in $(\alpha_i,\tau_i)$ can be related to group labels $g(i)$ and covariates $X_i$. Observed outcomes are defined as values of this function at the observed treatments: 
\begin{equation*}
Y_i = Y_i(W_i).
\end{equation*}
We focus on this model for its simplicity, an extension to a more general nonlinear model can be designed in the same way.

Our assumptions remain the same and imply the following unconfoundedness restriction:
\begin{equation*}
    (\alpha_i, \tau_i) \indep W_i | X_i, \osi.
\end{equation*}
Using this restriction, we can consider a general double robust strategy, where we estimate the model for $W_i$ and $Y_i$ using $X_i, \osi$ as regressors (\textit{e.g.}, \citet{chernozhukov2018double}). In particular, let $\hat m_{w,i,g}\equiv \hat m_{w,g}(X_i,\osi)$, and $\hat m_{y,i,g}\equiv \hat m_{y,g}(X_i,\osi)$ be the corresponding estimators (with cross-fitting) for the conditional means of $W_i$ and $Y_i$. Define the residuals:
\begin{equation*}
    \tilde Y_i \equiv Y_i -\hat m_{y,i,g}, \quad \tilde W_i \equiv W_i -\hat m_{w,i,g},
\end{equation*}
and consider the following regression:
\begin{equation}\label{eq:reg_nb_tr}
    \tilde Y_i = \alpha + \tau \tilde W_i + u_i.
\end{equation}
Let $\hat \tau$ be the OLS estimator for $\tau$. Using standard techniques one can show that $\tau$ is a weighted average of $\tau_i$, where the weights are proportional to the variance of $W_i$ conditional on $X_i,\osi$.

Regression (\ref{eq:reg_nb_tr}) is a direct generalization of the strategy based on (\ref{een}). To see this observe that representation (\ref{twee}) does not rely on $W_i$ being binary and can be seen as a particular implementation of (\ref{eq:reg_nb_tr}). Given the wide use of (\ref{een}) in applied work, we view (\ref{eq:reg_nb_tr}) as a natural extension of the standard regression with fixed effects with non-binary treatments. 

This discussion demonstrates that, practically, for our strategy, there are few differences between binary and non-binary treatments. Of course, the key restriction here is Assumption \ref{as:exp_family} that models the distribution of $(W_i,X_i)$ conditional on $L_{g(i)}$. With binary treatments, the group average $\ow_g$ is the only function of $W_i$ we can consider. Even with non-binary treatments, $\ow_g$ is the only function that is effectively used in (\ref{een}). The choice set is much larger with general treatments, and one can use other moments, for example, $\overline {W^2}$. The economic model discussed in Section \ref{sec:discussion} motivates using such objects and suggests a general way of building a  balancing score.

\section{Illustrations}
\subsection{Empirical example}
To illustrate our method, we use the data from \citet*{lacetera2012will}\nocite{data,data_2}, which focuses on the effects of various economic incentives on blood donations. To answer this question, the authors leverage a dataset on blood drives -- visits to particular locations -- conducted by the American Red Cross (ARC).For each drive, the authors have information on various outcomes, such as the number of present potential donors, blood units collected, and the number of people turned down. Each blood drive has multiple characteristics, such as the time of the year when it was conducted, its host, weather conditions on that day, and the way it was organized. Additionally, the data includes whether a reward for donation was offered, its type, and cost. 

We observe multiple drives with the same host with variation in incentives over them. As a result, the dataset has a grouped structure with unit-level variation in outcomes and treatments: each location defines a group $g$, with a blood drive to this location being a unit-level observation $i$. We use $W_i\in\{0,1\}$ for the absence or presence of economic incentive, covariates $X_i$ to describe the date of the drive and weather conditions on that day and $Y_i$ for the outcome of interest -- the number of collected blood units.

Assumption \ref{as:unconf_cl_un} is natural in this setup: appealing to institutional knowledge, the authors argue that conditional on the drive characteristics (date), the assignment of incentives is close to being random. Still, the probability of treatment can vary systematically over hosts, and thus conditioning on them -- $L_{g(i)}$ in our notation -- is crucial to claim unconfoundedness. We can expect hosts to differ in their fundamentals -- a potential number of donors and their elasticity with respect to economic incentives, relationship with ARC managers, and organizational abilities -- and $U_g$ captures these characteristics. Importantly, proximity in $U_g$ space might not necessarily be connected to the proximity in observed geographic locations. Our approach allows researchers to capture this unobserved proximity through  balancing scores. 

Finally, we need to specify the set of  balancing scores $S(W_i, X_i)$ that would justify Assumption \ref{as:exp_family}. In the paper, the authors employ specification (\ref{een}) which is equivalent to (\ref{twee}), and thus implicitly use $S(W_i, X_i) = (W_i, X_i)$. Following their analysis, we include $X_i$ and $W_i$ but also use interactions of functions of $X_i$ -- indicators for a particular season -- with $W_i$. If specific locations have strong seasonal patterns, we can expect ARC managers to react by adjusting incentives. Interactions  then help us to control for this additional selection channel.  

We start with $\numcl  = 2491$ unique hosts $N = 14029$ blood drives. We use temperature, amount of rain and snow, indicators for weekends, and seasons as characteristics $X_i$. For the interactions with $W_i$ we use season indicators. The dimension of  the balancing score is thus equal to $11$. Using constructed $(X_i, \osi)$, we estimate a propensity score model by random forest (\citet*{breiman2001random}) and keep units with estimated probability in $[0.05, 0.95]$. By being very flexible, random forest prediction allows us to drop units from groups without variation in treatment or those where treatment is perfectly predictable by  balancing scores.  We use this to define the set $\mathbb{A}$, with $\ova = 0.45$. We then use linear regression to estimate conditional mean $\mu(W_i, X_i, \osi)$, which is equivalent to (\ref{een}). Given the relatively small dimension of covariates, we do not use cross-fitting. 

\begin{table}[t]
\begin{center}
\begin{tabular}{|l|cccc|}
\hline
& Simple & F.E & IPW & Doubly Robust \\
\hline 
 \hline
Estimate &5.27 & 4.12 & 2.79 & 3.42  \\ 
  S.E.   &0.62  &  0.32& 0.64 & 0.44 \\ 
  \hline 
\end{tabular}
\caption{ $\numcl = 2,491$, $N = 14029$, $\ova = 0.45$. The first column corresponds to the OLS estimator without fixed effects, and the second one -- to the OLS estimator with fixed effects. The third column corresponds to the IPW estimator using the propensity score estimated by random forest, restricted to groups where the estimated score lies in $[0.05,0.95]$. The final column corresponds to the doubly robust estimator with the same propensity score model. The standard errors are robust to heteroscedasticity and within-group correlation. This simulation is conducted using x86-64 architecture. The results we get using Apple M1 Pro chips differ slightly. Using the latter, we get for the IPW estimator 2.58 with a s.e. of 0.65, and for the Doubly Robust estimator 3.50 a s.e. of 0.45. These results are reported as a part of the replication package.
  \label{table:est_res}}
 \end{center}
  \end{table}		

We report our results in Table \ref{table:est_res}, where we compare four different estimators: the OLS estimator for (\ref{een}) without group-level fixed effects, the fixed effects estimator $\htauf$, an inverse-propensity weights (IPW) estimator, and, finally, $\hat \tau_{\rm GME}$. The IPW estimator is defined analogously to the doubly robust estimator (\ref{eq:est_def}) but uses $\hat \mu_{g}(\cdot) \equiv 0$. In all cases, we report heteroscedasticity-robust standard errors corrected for within-group correlation. 

The results are similar across specifications, with the doubly robust estimator lying between the fixed effects estimator and the IPW estimator.  Balancing scores play an important role in predicting the treatment indicators; see Figure \ref{fig:p_model} in Appendix \ref{ap:det}. There are also efficiency gains between the IPW and the doubly robust estimator, with $\hat \tau_{\rm GME}$ having a smaller variance.

Similar to the original estimator used in the paper, our estimator $\hat \tau_{\rm GME}$ relies on (\ref{eq:sutva}). In particular, we do not allow for spillovers across groups. In \citet{lacetera2012will}, the authors argue that this assumption might be too restrictive since the supply of blood in neighboring locations can be interdependent. A possible solution for this problem is to aggregate the data at a higher geographical level and proceed as before, but it would be interesting to develop an alternative approach that uses the original data. Since this issue is orthogonal to the subject of the current paper, we do not attempt this development. 

\subsection{Monte-Carlo Experiments}

We base our simulation results on the dataset described above. In particular, we estimate a separate, host-level propensity score logit model for each location with four regressors: intercept and three seasonal indicators. We drop the hosts where these regressors are collinear, which leaves us with 616 hosts. We then combine the estimated coefficients into 20 groups using the $K$-means algorithm. This produces a model for $e(X_i, \oui)$, where four-dimensional $\oui$ represents the estimated coefficients for one of the 20 groups. We sample $M=1200$ hosts (with replacement) out of the original ones, keeping all characteristics of the blood drives and their total number fixed, and generate treatment assignments for each host and blood drive using constructed $e(X_i, \oui)$. We keep observed outcomes, which is equivalent to assuming a zero treatment effect for each unit. As a result, this simulation has two sources of randomness: sampled hosts and assignments, while both covariates and outcomes are kept fixed. In particular, this means that the simulation neither satisfies Assumption \ref{as:exp_family} with a simple  balancing score nor the regression model (\ref{een}), thus making the comparison between the two methods ambiguous ex-ante. 

We simulate this model 100 times, and for each simulation, we compute two estimators: the linear regression (\ref{een}) and the doubly robust estimator (\ref{eq:dr_est}). In the regression, we use all available covariates (the same as in the empirical exercise of the previous section), and for the propensity score, we use the standard implementation of random forest with 11-dimensional  balancing scores (the same as before). We keep observations with estimated propensity scores between 0.1 and 0.9, treating this set as $\mathbb{A}$. We use (\ref{eq:dr_var}) to estimate the asymptotic variance. We report the results in Table \ref{table:sim_res}. One can see that the doubly robust estimator outperforms the linear regression both in terms of bias and root-mean-square error (RMSE). The bias, however, is not negligible, which is partly a manifestation of misspecification and partly a sample size issue. The average over simulations of the estimated standard errors of $\hat \tau_{\rm GME}$ is equal to $0.58$, while the standard deviation of $\hat \tau_{\rm GME}$ over simulations is equal to $0.54$. This suggests that the variance estimator (\ref{eq:dr_var}) performs well in this simulation.

\begin{table}[t]
\begin{center}
\begin{tabular}{|l|rr|}
  \hline
 & Linear FE & Doubly Robust \\ 
  \hline
Bias & 0.79 & 0.38 \\ 
RMSE & 0.82 & 0.66 \\ 
   \hline
\end{tabular}
\caption{Sample size $\numcl = 1200$, $100$ simulations. The first column corresponds to the standard fixed effects estimator, and the second to the doubly robust estimator. We report average bias over simulations and root-mean-square error.
  \label{table:sim_res}}
 \end{center}
  \end{table}	
  
These results illustrate that our estimator, while not ideal, outperforms the conventional benchmark in the realistic data-driven simulation. We use standard off-the-shelf prediction methods widely available in various statistical packages to construct the estimator. More targeted balancing estimators analogous to those available for cross-sectional data (\textit{e.g.}, \citet*{athey2018approximate,chernozhukov2018adouble, hirshberg2021augmented}) might achieve even better results.

\section{Conclusion}

In this work, we proposed a new approach to identification and estimation in observational studies with unobserved group-level heterogeneity.  The identification argument is based on the combination of random effects and exponential family assumptions. Given this structure, we can identify a specific average treatment effect even in cases where the observed number of units per group is small. From the operational point of view, our approach allows researchers to utilize all the recently developed machinery from the standard observational studies. In particular, we generalize the doubly robust estimator and prove its consistency and asymptotic normality under common high-level assumptions.

\newpage
\bibliographystyle{plainnat}
\bibliography{references}
\newpage
\begin{appendices}
\setcounter{lemma}{0}
\renewcommand{\thelemma}{\Alph{section}\arabic{lemma}}
\setcounter{prop}{0}
\renewcommand{\theprop}{\Alph{section}\arabic{prop}}
\setcounter{corollary}{0}
\renewcommand{\thecorollary}{\Alph{section}\arabic{corollary}}

\section{Identification results}\label{app:ident}
\footnotesize

In the main text, we assumed that $N_g$ is constant, but our results hold more generally. Formally, redefine $\ou_g \equiv (\eta(L_g),N_g)$, and $\os_g\equiv(\sum_{i:g(i) = g}S(W_i,X_i)/N_g,N_g)$, use $\os_{1g}\equiv \sum_{i:g(i) = g}S(W_i,X_i)/N_g$, and assume that $N_g = N(L_g)\ge2$ for some function $N$. We then prove the results using this definition. In the main text, there is no need to include $N_g$ in $\os_g$ as it is constant by Assumption \ref{as:random_sampling}.\\

\noindent \textbf{Proof of Proposition \ref{prop:unc_in_general}}:
\begin{proof}
Since distribution of $(W_i,X_i)$ conditional on $ \mathbb{P}_{g(i)}$ and $L_{g(i)}$ is equal to $ \mathbb{P}_{g(i)}$ and thus does not depend on $L_{g(i)}$, we have a conditional independence restriction:
\begin{equation}\label{eq:pr_0}
    (W_i,X_i) \indep L_{g(i)} | \mathbb{P}_{g(i)}.
\end{equation}
To prove the result, we first show the following condition:
\begin{equation}\label{eq:ad_ind}
    (Y_i(1),Y_i(0))\indep W_i| X_i, \mathbb{P}_{g(i)}, L_{g(i)}.
\end{equation}
Fix an arbitrary bounded function $f(y_1,y_0)$ and consider the following equalities:
\begin{multline}\label{eq:pr_1}
    \mathbb{E}[f(Y_{i}(1),Y_i(0))| W_i,X_i, \mathbb{P}_{g(i)}, L_{g(i)}] =  \mathbb{E}[f(Y_{i}(1),Y_i(0))| W_i,X_i, L_{g(i)}] =\\
    \mathbb{E}[f(Y_{i}(1),Y_i(0))| X_i, L_{g(i)}],
\end{multline}
where the first one follows from Assumption \ref{as:random_sampling} and the fact that $N_g = N(L_g)$, and the second -- from Assumption \ref{as:unconf_cl_un}. We thus have
\begin{multline}\label{eq:der}
    \mathbb{E}[W_i|Y_i(1), Y_i(0),X_i, \mathbb{P}_{g(i)}] =\\ \mathbb{E}[\mathbb{E}[W_i|Y_i(1), Y_i(0),X_i, \mathbb{P}_{g(i)}, L_{g(i)}]|Y_i(1), Y_i(0),X_i, \mathbb{P}_{g(i)}] = \\
\mathbb{E}[\mathbb{E}[W_i|X_i, \mathbb{P}_{g(i)}, L_{g(i)}]|Y_i(1), Y_i(0),X_i, \mathbb{P}_{g(i)}]=\\
\mathbb{E}[\mathbb{E}[W_i|X_i, \mathbb{P}_{g(i)}]|Y_i(1), Y_i(0),X_i,  \mathbb{P}_{g(i)}] = \mathbb{E}[W_i|X_i, \mathbb{P}_{g(i)}],
\end{multline}
where the second equality follows from (\ref{eq:ad_ind}), and the third one -- from (\ref{eq:pr_0}), since it implies that $W_i \indep L_{g(i)} | X_i,\mathbb{P}_{g(i)}$. Since $W_i$ is binary (\ref{eq:der}) implies the result.\footnote{The proof for a non-binary case follows by applying the same chain of equalities to a bounded function of $W_i$.}
\end{proof}

\noindent\textbf{Proof of Lemma \ref{lemma:re}}:
\begin{proof}
The first part of the lemma follows by definition of $\oui$ and Assumptions \ref{as:random_sampling}, \ref{as:exp_family}. To demonstrate this, we view $\left(\{W_i, X_i\}_{i\in g}, N_g\right)$ as an element of $\left(\{0,1\}\times \mathbb{X}\right)^{\infty}\times \mathbb{N}$, where $\mathbb{N}$ is the set of natural numbers and fix an arbitrary bounded function $f_0 \in B(\left(\{0,1\}\times\mathbb{X}\right)^{\infty}\times\mathbb{N})$, where $B(\left(\{0,1\}\times\mathbb{X}\right)^{\infty}\times\mathbb{N})$ is a set of bounded functions defined on $\left(\{0,1\}\times\mathbb{X}\right)^{\infty}\times\mathbb{N}$.\footnote{Formally, we define a mapping that takes an element $\left(\{W_i, X_i\}_{i\in g}, N_g\right)$ and maps it into $\left(\{W_i, X_i\}_{i\in g}\times\{0,x_0\}^{\infty}, N_g\right)$, where $x_0$ is a fixed element of $\mathbb{X}$.} For this function $f_0$, we define the following quantity (we use the fact that $\ou_g$ is a deterministic function of $L_g$):
\begin{multline}
   v(f_0, L_g, \ou_g)\equiv \mathbb{E}\left[f_0\left(\{W_i,X_i\}_{g(i) =g},N_g\right) | L_{g}, \ou_g\right] = \\
   \exp^{N_g}\left(\eta_0(L_g)\right) \int f_0\left(\{w_i,x_i\}_{i=1}^{N_{g}},N_g\right) \prod_{i=1}^{N_g}h(w_i, x_i)\exp^{N_g}\left(\ou_g^\top \frac{1}{N_g}\sum_{i=1}^{N_g}S(w_i,x_i)\right) \prod_{i=1}^{N_g}(d\mu(w_i,x_i)) \equiv \\
    \exp^{N_g}\left(\eta_0(L_g)\right) v_0(f_0,\ou_g).
\end{multline}
Since $ v(1, L_g, \ou_g) =1$ and $\frac{v(f_0,L_g, \ou_g)}{v(1,L_g, \ou_g)} =  \frac{v_0(f_0,\ou_g)}{v_0(1,\ou_g)}$ it follows that $v(f_0, L_g, \ou_g)$ does not depend on $L_g$ proving the first part. The following equalities imply the second part:
\begin{multline}
    \mathbb{E}[W_i|Y_i(1),Y_i(0),X_i,\oui] = 
    \mathbb{E}[\mathbb{E}[W_i|Y_i(1),Y_i(0),L_{g(i)},X_i,\oui] | Y_i(1),Y_i(0),X_i,\oui] =\\
    \mathbb{E}[\mathbb{E}[W_i|L_{g(i)},X_i,\oui] | Y_i(1),Y_i(0),X_i,\oui] = \\
     \mathbb{E}[\mathbb{E}[W_i|X_i,\oui] | Y_i(1),Y_i(0),X_i,\oui] = \mathbb{E}[W_i|X_i,\oui]
\end{multline}
where the second equality follows from Assumption \ref{as:unconf_cl_un}, and the third inequality follows from the first part of the lemma, thus proving the result.
\end{proof}

\noindent\textbf{Proof of Theorem \ref{th:unc_with_exp_family}}:
\begin{proof}
Assumptions \ref{as:random_sampling}, and \ref{as:exp_family} imply the following:
\begin{equation}\label{eq:exp_ind}
    \ou_g \indep \{W_i,X_i\}_{i: g(i)= g} | \os_g.
\end{equation}
Formally, using the same approach as in the proof of Lemma \ref{lemma:re} and its first part, we get the following:
\begin{multline}
   \tilde v(f_0, \os_g, \ou_g)\equiv \mathbb{E}\left[f_0\left(\{W_i,X_i\}_{g(i) =g},N_g\right) | \os_g, \ou_g\right] =
   \mathbb{E}\left[f_0\left(\{W_i,X_i\}_{g(i) =g},N_g\right) | \os_g, \ou_g, L_g\right] = \\
   \exp^{N_g}\left(\eta_0(L_g)\right) \exp^{N_g}\left(\ou_g^\top S_{1g}\right)\int_{ \frac{1}{N_g}\sum_{i=1}^{N_g}S(w_i,x_i) = \os_{1g}} f_0\left(\{w_i,x_i\}_{i=1}^{N_{g}},N_g\right) \prod_{i=1}^{N_g}h(w_i, x_i) \prod_{i=1}^{N_g}(d\mu(w_i,x_i)) \equiv \\
    \exp^{N_g}\left(\eta_0(L_g)\right) \exp^{N_g}\left(\ou_g^\top S_{1g}\right) \tilde v_0(f_0,\os_g).
\end{multline}
Again, using the fact that $\tilde v(1, \os_g, \ou_g) =1$ we get that $\tilde v(f_0, \os_g, \ou_g)$ does not depend on $\ou_g$. Using this result, we get for an arbitrary bounded function $f(y_1,y_0)$:
\begin{multline}
    \mathbb{E}[f(Y_{i}(1),Y_{i}(0))|W_i, X_i, \osi] = \\
    \mathbb{E}[\mathbb{E}[f(Y_{i}(1),Y_{i}(0))|W_i, X_i, \osi,\oui]|W_i, X_i,\osi] = \\
    \mathbb{E}[\mathbb{E}[f(Y_{i}(1),Y_{i}(0))|X_i,\oui]|W_i, X_i,\osi]  =  \mathbb{E}[f(Y_{i}(1),Y_{i}(0))| X_i,\osi],
\end{multline}
where the second equality follows from Lemma \ref{lemma:re} and Assumption \ref{as:random_sampling}, and the last equality follows from (\ref{eq:exp_ind}). 
\end{proof}

\section{Inference results}\label{ap:sem_inf}
\noindent{\bf{Notation:}} We are using standard notation from the empirical processes literature adapted to our setting. For any \textbf{group-level} random vector $X_g$ define $\mathbb{E}_{\numcl}(X_g) \equiv  \frac{1}{\numcl} \sum_{g=1}^\numcl X_g$. Define $B_i = (X_i,\overline{S}_{g(i)})$ and $D_i \equiv  (W_i,B_i)$ and the following objects:
\begin{equation}
\begin{aligned}
    &\mu_i\equiv \mathbb{E}[Y_i|D_i],\quad \hat\mu_{i,g}\equiv  \hat \mu_g(D_i),\quad\mu_{w,i}\equiv \mathbb{E}[Y_i|W_i= w, B_i], \quad \hat\mu_{w,i,g}\equiv \hat \mu_g(w,B_i),\\
     &e_i\equiv \mathbb{E}[W_i|B_i], \quad \hat e_{i,g}\equiv \hat e_g(B_i).
\end{aligned}
\end{equation}

\textbf{Proof of Theorem \ref{th:inference}}:
\begin{proof}
We separate $\rho_g(m,p)$ into two parts:
\begin{multline}
\rho_g(m,p) =  \sum_{i: g(i)= g}\frac{1}{N_g}A_i \left(m(1,B_i) + \frac{W_i}{p(B_i)}(Y_i -m(1,B_i))\right) -\\
 \sum_{i: g(i)= g}\frac{1}{N_g}A_i \left(m(0,B_i) + \frac{1-W_i}{1- p(B_i)}(Y_i - m(0,B_i))\right) = \rho_{1g}( m, p) -  \rho_{0g}(m,p).
\end{multline}
We analyze $ \rho_{1g}(\hat \mu_g,\hat e_g)$; the other term is analyzed in the same way. We have the following expansion:
\begin{multline}
    \rho_{1g}(\hat \mu_{1,g},\hat e_g) =\rho_{1g}(\mu, e) + 
  \sum_{i: g(i)= g}\frac{1}{N_g}A_i \left(\left(\hat  \mu_{1,i,g} -\mu_{1,i}\right)\left(1 -  \frac{W_i}{ e_i}\right)\right) + \\
 \sum_{i: g(i)= g}\frac{1}{N_g}A_i \left(\left(\hat  \mu_{i,g} -\mu_{i}\right)W_i\left(\frac{\hat e_{i,g}- e_i}{ e_i\hat e_{i,g}} \right)\right)+
\sum_{i: g(i)= g}\frac{1}{N_g}A_i(Y_i -  \mu_{i})W_i\left(\frac{1}{ \hat e_{i,g}} -  \frac{1}{ e_i}\right) =\\
\rho_{1g}(\mu, e)+ R_{1g} + R_{2g} + R_{3g}.
\end{multline}
We analyze the last three terms separately. By cross-fitting, $\hat \mu_{k,i,g}$ is constructed without using the data from group $g$ and since $\mathbb{E}[A_i(1-\frac{W_i}{e_i})|B_i] = 0$, for the first term we have:
\begin{equation}
    \mathbb{E}\left[\mathbb{E}_{\numcl}R_{1g}\right] = 0.
\end{equation}
We then bound its second moment:
\begin{multline}
\mathbb{E}[\left(\mathbb{E}_{\numcl}R_{1g}\right)^2] \le \frac{L}{\numcl} \mathbb{E}\left[\mathbb{E}_{\numcl}\sum_{i: g(i)= g}\frac{1}{N_g}A_i \left(\hat  \mu_{1,i,g} -\mu_{1,i}\right)^2\left(1 -  \frac{W_i}{ e_i}\right)^2\right] \le \\
\frac{1-\eta}{\eta}\frac{L}{\numcl} \mathbb{E}\left[\mathbb{E}_{\numcl}\sum_{i: g(i)= g}\frac{1}{N_g}A_i \left(\hat  \mu_{1,i,g} -\mu_{1,i}\right)^2\right] = o\left(\frac{1}{\numcl}\right).
\end{multline}
From Chebyshev inequality it then follows  $\mathbb{E}_{\numcl}R_{1g} = o_p\left(\frac{1}{\sqrt{\numcl}}\right)$. The second term is bounded in the following way:
\begin{multline}
   \left|\mathbb{E}_{\numcl}R_{2g}\right| \le \frac{1}{\eta^2}\left|\mathbb{E}_{\numcl}\sum_{i: g(i)= g}\frac{1}{N_g}A_i \left(\hat  \mu_{i,g} -\mu_i\right)\left(\hat e_{i,g}- e_i \right)\right| \le\\
   \frac{1}{\eta^2}\sqrt{\mathbb{E}_{\numcl}\sum_{i: g(i)= g}\frac{1}{N_g}A_i \left(\hat  \mu_{i,g} -\mu_i\right)^2}\sqrt{\mathbb{E}_{\numcl}\sum_{i: g(i)= g}\frac{1}{N_g}A_i\left(\hat e_{i,g}- e_i \right)^2} =
   o_p\left(\frac{1}{\sqrt{\numcl}}\right). 
\end{multline}
Similarly to the first term, because $e_{i,g}$ is constructed without using data from the group $g$ and $\mathbb{E}[Y_i - \mu_i|D_i] = 0$, we have for the third term
\begin{equation}
     \mathbb{E}\left[\mathbb{E}_{\numcl}R_{3g}\right] = 0,
\end{equation}
and bounded second moment:
\begin{multline}
    \mathbb{E}[\left(\mathbb{E}_{\numcl}R_{3g}\right)^2] \le \frac{L}{\numcl}\frac{16}{\eta^4} \mathbb{E}\left[\mathbb{E}_{\numcl}\sum_{i: g(i)= g}\frac{1}{N_g}A_i(Y_i -  \mu_i)^2\left(\hat e_{i,g} - e_i\right)^2 \right] \le\\
    \frac{L}{\numcl}\frac{16K}{\eta^4} \mathbb{E}\left[\mathbb{E}_{\numcl}\sum_{i: g(i)= g}\frac{1}{N_g}A_i\left(\hat e_{i,g} - e_i\right)^2 \right] = o\left(\frac{1}{\numcl}\right).
\end{multline}
From Chebyshev inequality it follows that $\mathbb{E}_{\numcl}R_{3g} = o_p\left(\frac{1}{\sqrt{\numcl}}\right)$. Combining bounds for $R_{1g}, R_{2g}, R_{3g}$ together we get 
\begin{equation}
\mathbb{E}_{\numcl} \rho_{1g}(\hat \mu_g, \hat e_g) = \mathbb{E}_{\numcl} \rho_{1g}(\mu, e) + o_p\left(\frac{1}{\sqrt{\numcl}}\right).
\end{equation}
Repeating the same arguments for $\hat \rho_{0g}$ and combining the results we get
\begin{equation}
    \mathbb{E}_{\numcl}  \rho_g(\hat \mu_g,\hat e_g) =  \mathbb{E}_{\numcl}\sum_{i: g(i)= g}\frac{1}{N_g}A_i \tau(B_i)  + \mathbb{E}_{\numcl} \xi_g + o_p\left(\frac{1}{\sqrt{\numcl}}\right).
\end{equation}
By the LLN we have $\mathbb{E}_{\numcl}\sum_{i: g(i)= g}\frac{A_i}{N_g} = \mathbb{E}[A_i] + o_p(1)$ and thus we have the following:
\begin{equation}\label{eq:final_rep}
   \hat \tau_{dr}  - \frac{\mathbb{E}_{\numcl}\sum_{i: g(i)= g}\frac{1}{N_g}A_i \tau(B_i)}{\mathbb{E}_{\numcl}\sum_{i: g(i)= g}\frac{A_i}{N_g} } = \frac{\mathbb{E}_{\numcl} \xi_g}{ \mathbb{E}[A_i]} + o_p\left(\frac{1}{\sqrt{\numcl}}\right).
\end{equation}
Since $\mathbb{E}[\xi_g] = 0$ and $\xi_g$ are i.i.d., the two needed results follow from (\ref{eq:final_rep}). First, by the LLN, we get consistency:
\begin{equation}
    \hat \tau_{dr} = \frac{\mathbb{E}_{\numcl}\sum_{i: g(i)= g}\frac{1}{N_g}A_i \tau(B_i)}{\mathbb{E}_{\numcl}\sum_{i: g(i)= g}\frac{A_i}{N_g} } + o_p(1) = \mathbb{E}[\tau(B_i)|B_i \in \mathbb{A}] + o_p(1).
\end{equation}
Second, by the CLT, we get the asymptotic distribution of our estimator around the conditional estimand:
\begin{equation}
     \sqrt{\numcl}\left(\hat \tau_{dr} - \frac{\mathbb{E}_{\numcl}\sum_{i: g(i)= g}\frac{1}{N_g}A_i \tau(B_i)}{\mathbb{E}_{\numcl}\sum_{i: g(i)= g}\frac{A_i}{N_g} } \right) \Rightarrow \mathcal{N}\left(0, \mathbb{V}\right),
\end{equation}
where $\mathbb{V} \equiv \frac{\mathbb{E}[\xi_g^2]}{\mathbb{E}^2[A_i]}$.
\end{proof}

\noindent{\bf Proof of Theorem \ref{th:var_cons}}: 
\begin{proof}
Similarly to the previous proof we can divide $\xi_g$ into two parts $\xi_{1g}$ and $\xi_{0g}$. We will analyze $\xi_{1g}$, the analysis for $\xi_{0g}$ is the same. We have the following decomposition:
\begin{multline}
\hat \xi_{1g} - \xi_{1g} = \sum_{i: g(i)= g}\frac{1}{N_g}A_i \left(\left(\mu_i- \hat  \mu_{i,g} \right)\frac{W_i}{\hat e_{i,g}}\right) + \\
\sum_{i: g(i)= g}\frac{1}{N_g}A_i (Y_i -  \mu_i)W_i\left(\frac{1}{ \hat e_{i,g}} -  \frac{1}{ e_i}\right) = R_{11g} + R_{12g}  
\end{multline}
For the first term, we have the following bound:
\begin{equation}
\left(\mathbb{P}_nR_{11g}\right)^2 \le \mathbb{E}_{\numcl}\frac{1}{N_g}\sum_{i: g(i)= g} A_i  \left(\left(\mu_i- \hat  \mu_{i,g} \right)^2\frac{W_i}{\hat e^2_{i,g}}\right) \le \frac{4}{\eta^2} \mathbb{E}_{\numcl}\frac{1}{N_g}\sum_{i: g(i)= g} A_i  \left(\left(\mu_{1,i}- \hat  \mu_{1,i,g} \right)^2\right)  = o_p(1),
\end{equation}
where the last inequality follows from part (i) of Assumption \ref{as:cons_prop_score}, and the last equality follows from part (ii).

For the second term, we have the following bound:
\begin{multline}
\left(\mathbb{E}_{\numcl}R_{12g}\right)^2 \le \mathbb{E}_{\numcl} \frac{1}{N_g} \sum_{i: g(i)= g}A_iW_i\varepsilon_i^2 \frac{(\hat e_{i,g} - e_i)^2}{\hat e^2_{i,g} e^2_i} \le 
\sqrt{\left(\mathbb{E}_{\numcl} \frac{1}{N_g} \sum_{i: g(i)= g}A_i \frac{(\hat e_{i,g} - e_i)^4}{\hat e^4_{i,g} e^4_i}\right)\left(\mathbb{E}_{\numcl} \frac{\sum_{i: g(i)= g}A_iW_i\varepsilon_i^4}{N_g} \right)}\le \\
 \frac{4}{\eta^4}\sqrt{\left(\mathbb{E}_{\numcl} \frac{1}{N_g} \sum_{i: g(i)= g}A_i(\hat e_{i,g} - e_i)^2\right) \left(\mathbb{E}_{\numcl} \frac{1}{N_g} \sum_{i: g(i)= g}A_iW_i\varepsilon_i^4\right)}=
 o_p(1) O_p(1) = o_p(1),
\end{multline}
where for the last inequality, we used $\frac{(\hat e_{i,g} - e_i)^4}{\hat e^4_{i,g} e^4_i} \le  \frac{(\hat e_{i,g} - e_i)^2}{\hat e^4_{i,g} e^4_i} \le \frac{16}{\eta^8}(\hat e_{i,g} - e_i)^2$, which follows from part (i) of Assumption \ref{as:cons_prop_score}. For the last equality we used $\mathbb{E}_{\numcl} \frac{1}{N_g} \sum_{i: g(i)= g}A_iW_i\varepsilon_i^4 \le \mathbb{E}_{\numcl} \frac{1}{N_g} \sum_{i: g(i)= g}\varepsilon_i(1)^4 = \frac{1}{N} \sum_{i=1}^N \varepsilon_i(1)^4 = O_p(1)$ by Assumption \ref{as: moment_cond}.

Putting these results together, we have the following:
\begin{multline}
\mathbb{E}_{\numcl}(\hat \xi_{1g} + \hat \xi_{2g})^2  -  \mathbb{E}_{\numcl}(\xi_{1g} + \xi_{2g})^2 = \mathbb{E}_{\numcl}(\xi_{1g} + \xi_{2g} + R_{11g}+R_{12g} + R_{01g}+R_{02g})^2  - \mathbb{E}_{\numcl}(\xi_{1g} + \xi_{2g})^2 = \\
 \mathbb{E}_{\numcl}(\xi_{1g} + \xi_{2g})(R_{11g}+R_{12g} + R_{01g}+R_{02g}) +   \mathbb{E}_{\numcl}(R_{11g}+R_{12g} + R_{01g}+R_{02g})^2 \le \\
 \sqrt{ \mathbb{E}_{\numcl}(\xi_{1g} + \xi_{2g})^2 4 \mathbb{E}_{\numcl}(R^2_{11g}+R^2_{12g} + R^2_{01g}+R^2_{02g})} +  \mathbb{E}_{\numcl}(R^2_{11g}+R^2_{12g} + R^2_{01g}+R^2_{02g}) = \\
 \sqrt{O_p(1) o_p(1)} + o_p(1) = o_p(1)
\end{multline}
This argument also implies that $\mathbb{E}_{\numcl}(\hat \xi_{1g}) =  \mathbb{E}_{\numcl}(\xi_{1g})= o_p(1)$ and thus we have the final result:
\begin{multline}
\frac{1}{(\overline A)^2}\left(\mathbb{E}_{\numcl}(\hat \xi_{1g}  + \hat \xi_{2g})^2 - \left(\mathbb{E}_{\numcl}(\hat \xi_{1g}  + \hat \xi_{2g})\right)^2\right)- \frac{1}{\mathbb{E}^2[A_i]}\mathbb{E}_{\numcl}(\xi_{1g} +\xi_{2g})^2 = \\
\frac{1}{(\overline A)^2}\left( \mathbb{E}_{\numcl}(\hat \xi_{1g} + \hat \xi_{2g})^2 - \mathbb{E}_{\numcl}(\xi_{1g} +\xi_{2g})^2 \right) + \left(\frac{1}{(\overline A)^2} - \frac{1}{\mathbb{E}^2[A_i]}\right)  \mathbb{E}_{\numcl}(\xi_{1g} +\xi_{2g})^2   + O_p(1)o_p(1)=\\
O_p(1) o_p(1) + o_p(1) O_p(1) + O_p(1)o_p(1) = o_p(1)
\end{multline}
\end{proof}

\newpage
\section{Estimation Details}\label{ap:det}

\begin{figure}[h]
    \centering
    \includegraphics[scale = 0.75]{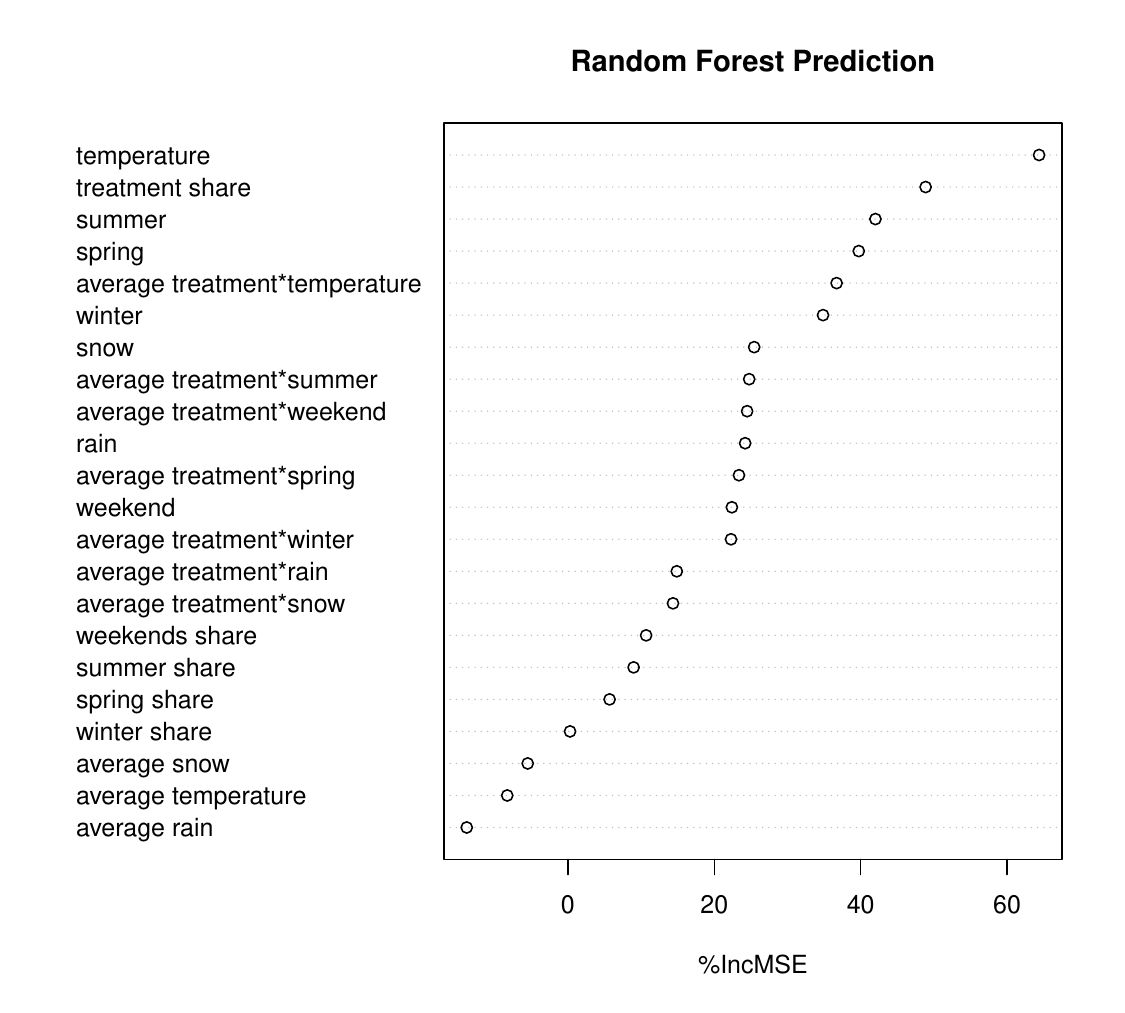}
    \caption{$\numcl = 2491$, $N = 14029$. Each row is the \% increase in the mean-squared error if the corresponding factor was dropped from the propensity score model estimated by the random forest.}
    \label{fig:p_model}
\end{figure}

\end{appendices}

\end{document}